%% file: corr-1.tex
\documentclass[draft,twoside]{article}
\input preamble.tex

\numberwithin{equation}{section} 

\numberwithin{equation}{section}

\begin{document}
\title{\bf Correlation Functions of Ensembles of Asymmetric Real
  Matrices}  
\author{\sc A.~Borodin, C.D.~Sinclair}
\maketitle

\begin{abstract}
We give a closed form for the correlation functions of
ensembles of asymmetric real matrices in terms of the Pfaffian of an
antisymmetric matrix formed from a $2 \times 2$ matrix kernel
associated to the ensemble.  We also derive closed forms for the
matrix kernel and correlation functions for Ginibre's real ensemble.  
\end{abstract}

\noindent Difficulties arise in the study of ensembles of asymmetric random
matrices which do not occur in the more commonly studied ensembles of
random matrices ({\it e.g.} Hermitian, unitary, complex asymmetric,
etc.). These difficulties stem from the basic fact that the
eigenvalues of real matrices can either be real or members of complex
conjugate pairs. Due to this additional complexity, our understanding
of ensembles of real asymmetric matrices lags far behind our
understanding of the Hermitian ensembles. In spite of these
complications many important quantities such as a closed form for
ensemble averages and correlation functions can be expressed using
analogous structures for their Hermitian counterparts.

Here we will show that the correlation functions for an ensemble of
asymmetric real matrices can be expressed as the Pfaffian of a block
matrix whose entries are expressed in terms of a $2 \times 2$ matrix
kernel associated to the ensemble. We find much inspiration from Tracy
and Widom's paper on correlation and cluster functions of Hermitian
and related ensembles \cite{MR1657844}. However, instead of using
properties of the Fredholm determinant to calculate the correlation
functions via the cluster functions, we use the notion of the Fredholm
Pfaffian to determine the correlation functions directly. A Pfaffian
analog of the Cauchy-Binet formula introduced by Rains [10] lies at
the heart of our proof. For completeness we will include Rains' proof
here. In place of the identities of de Bruijn \cite{MR0079647} used by
Tracy and Widom we will use an identity of the second author
\cite{sinclair-2007} to compute the correlation functions for
ensembles of asymmetric matrices. Rains' Cauchy-Binet formula has
applicability in a wider context than just the determination of the
correlation function for ensembles of asymmetric matrices, and we will
use it to give a simplified proof of the correlation function of
Hermitian ensembles of random matrices when $\beta = 1$ and $\beta=
4$.

Next we will apply our theory to Ginibre's ensemble of real asymmetric
matrices.  Matrices in this ensemble are square (say $N \times N$),
asymmetric with iid Gaussian entries.  By scaling the eigenvalues by
$\sqrt{N}$, as $N \rightarrow \infty$, the support of the eigenvalues
converges to the unit disk.  This scaling allows us to analyze the
asymptotics of the matrix kernel in the bulk, and we will discover a
new universality class of matrix models.  This will yield a closed
form for the limiting correlations of points in the bulk.  Two
different kernels emerge: one when expanding about points on the real
axis and another when expanding about point in the (non-real) complex
bulk.  In the latter case we find that the limiting (large $N$)
correlation functions of Ginibre's real ensemble in the bulk off the
real line are identical to the limiting correlation functions of
Ginibre's complex ensemble in the bulk.  

\vspace{.5cm}

\noindent{\it Acknowledgments.} We thank Percy Deift
for helpful early discussions and for introducing the authors to each
other.  The second author would also like to thank Brian Rider for
many helpful discussions regarding Ginibre's real ensemble.  Finally,
we would like to thank Eric Rains for allowing us to include his proof
of the the Pfaffian Cauchy-Binet formula (\ref{app:B}). 

\section{Point Processes on the Space of Eigenvalues of $\R^{N
    \times N}$}
\label{sec:point-proc-eigenv}
We start rather generally since the results in this manuscript can be
used to describe not only the statistics of eigenvalues of ensembles
of real matrices but also the statistics of roots of certain ensembles
of real polynomials.  We begin with random point processes
on two-component systems. 

Let $X$ be the set of finite multisets of the closed upper half plane
$\overline{H} \subset \C$.  An element $\xi \in X$ is called a {\em
  configuration}, and $X$ is called the {\em configuration space} of
$\overline{H}$.  Given a Borel set $A$ of $\overline{H}$, we define
the function $\mathcal{N}_A: X \rightarrow \Z^{\geq 0}$ by specifying
that $\mathcal{N}_A(\xi)$ be the cardinality (as a multiset) of $(\xi
\cap A)$. We define $\Sigma$ to be the sigma-algebra on $X$
generated by $\{ \mathcal{N}_A: A \subseteq \overline{H} \mbox{
  Borel}\}$.  A probability measure $\mathcal{P}$ defined on $\Sigma$
is called a {\em random point process} on $X$.

For each pair of non-negative integers $(L,M)$ we define $X_{L,M}$ to
be the subset of $X$ consisting of those configurations $\xi$ which
consist of exactly $L$ real points and $M$ points in the open upper
half plane $H$.  That is,
\[
X_{L,M} := \big\{ \xi \in X : \mathcal{N}_{\R}(\xi) = L \mbox{ and }
\mathcal{N}_H(\xi) = M \big\}.  
\]
Clearly, $X_{L,M}$ is measurable and $X$ can be written as the disjoint
union 
\[
X = \bigcup_{L \geq 0, M \geq 0} X_{L,M}.
\]
Given a point process $\mathcal{P}$ on $X$, we may define the measure
$\mathcal{P}_{L,M}$ on $X$ by $\mathcal{P}_{L,M}(B) := \mathcal{P}(B
\cap X_{L,M} )$ for each $B \in \Sigma$ . The measure
$\mathcal{P}_{L,M}$ induces a measure on $X_{L,M}$ (and we will also
use the symbol $\mathcal{P}_{L,M}$ for this measure). 

Given a matrix $\mathbf{Y} \in \R^{N \times N}$ there must be a pair of
non-negative integers $(L,M)$ with $L + 2M = N$ such that, counting
multiplicities, $\mathbf{Y}$ has $L$ real eigenvalues and $M$ non-real complex
conjugate pairs of eigenvalues.  By representing each pair of complex
conjugate eigenvalues by its representative in $H$, we may
identify all possible multisets of eigenvalues of 
matrices in $\R^{N \times N}$ with the disjoint union
\begin{equation}
\label{eq:1}
X_N := \bigcup_{(L,M) \atop L + 2M = N} X_{L,M}.
\end{equation}
Similarly, we may identify all possible multisets of roots of degree
$N$ real polynomials with this disjoint union.  Thus, when studying
the statistics of eigenvalues of ensembles of asymmetric real matrices
(respectively of the roots of degree $N$ real polynomials) we may
restrict ourselves to random point processes on $H$ which are
supported on the disjoint union given in (\ref{eq:1}).  That is, the
eigenvalue statistics of an ensemble of real matrices is determined by
a set of finite measures $\mathcal{P}_{L,M}$ on $X_{L,M}$ for each
pair of non-negative integers $(L,M)$ with $L + 2M = N$.  In this
situation we will say that $\mathcal{P}$ is a random point process on
$X_N$ associated to the family of finite measures $\{\mathcal{P}_{L,M}
: L + 2M = N\}$.  

From here forward we will assume that $L, M$ and $N$ are non-negative
integers such that $L + 2M = N$, and a sum indexed over $(L, M )$ will
be taken to be a sum over all pairs satisfying this
condition. Moreover, when we refer to a point process we will always
mean a point process on $X_N$.

\section{Point Processes on $X_N$  Associated to Weights} 
\label{sec:point-processes-xn}

In this section we will introduce an important class of point
processes associated to Borel measures on $\C$.  We will be
particularly interested in measures which are a sum of two mutually
singular measures: one of which is absolutely continuous with respect
to Lebesgue measure on $\R$ and one which is absolutely continuous
with respect to Lebesgue measure on $\C$.  The corresponding densities
with respect to Lebesgue measure will allow us to construct a {\em
  weight function} which uniquely determines the associated point
process on $X_N$.  Point processes of this sort arise in the study of
asymmetric random matrices and the range of multiplicative functions
on polynomials with real coefficients.

It will be useful to distinguish non-real complex numbers and we
set $\CnoR = \C \setminus \R$.   

We start rather generally by constructing measures on $X_{L,M}$ from
measures on  $\R^L \times \CnoR^M$.  The benefit in
doing this is that it  allows us to express important
quantities associated to point processes (averages, correlation
functions, etc.) as rather pedestrian integrals over $\R^L \times
\CnoR^M$.   To each $(\bs{\alpha}, \bs{\beta}) \in \R^L \times
\CnoR^M$ we associate a configuration in $X_{L,M}$ given by
\[
\{\bs{\alpha}, \bs{\beta}\} := \{\alpha_1, \ldots, \alpha_L,
\widehat{\beta}_1, \ldots, \widehat{\beta}_M\}, \quad \mbox{where}
\quad \widehat{\beta_m} 
= \{ \beta_m, \overline{\beta_m} \} \cap H.
\]
A given configuration $\xi \in X_{L,M}$ may correspond to several
vectors in $\R^L \times \CnoR^M$ and we will call $\{ (\bs{\alpha},
\bs{\beta} ) : \{ \bs{\alpha}, \bs{\beta} \} = \xi \}$ the set of {\em
  configuration vectors} of $\xi$. 

A function $F$ on $X_{L,M}$ induces a function on $\R^L \times
\CnoR^M$ specified by $(\bs{\alpha}, \bs{\beta}) \mapsto
F\{\bs{\alpha}, \bs{\beta}\}$, and given a measure $\nu_{L,M}$ on
$\R^L \times \CnoR^M$ there exists a unique measure
$\mathcal{P}_{L,M}$ on $X_{L,M}$ specified by demanding that
\begin{equation}
\label{eq:3}
\int_{X_{L,M}} F(\xi) \, d\mathcal{P}_{L,M}(\xi) = \frac{1}{L! M! 2^M}
\int_{\R^L \times 
  \CnoRsub^M} F\{ \bs{\alpha}, \bs{\beta} \} \,
d\nu_{L,M}(\bs{\alpha}, \bs{\beta}),
\end{equation}
for every $\Sigma$-measurable function $F$ on $X$.  The normalization
constant $L! M! 2^M$ arises since a generic element $\xi \in X_{L,M}$
corresponds to $L! M! 2^M$ configuration vectors.  By specifying
measures $\nu_{L,M}$ on $\R^L \times \CnoR^M$ for all pairs $(L,M)$
and normalizing so that the total measure of $X_N$ is 1, we define a
point process on $X_N$.    

A very important class of point processes arises when we demand that
the various $\nu_{L,M}$ are all related to a single measure on $\C$.
Given a measure $\nu_1$ on $\R$ and a measure $\nu_2$ on $\CnoR$, we
set $\nu$ to be the measure $\nu_1 + \nu_2$ on $\C$.  We will write
$\nu_L$ for the product measure of $\nu_1$ on $\R^L$ and $\nu_{2M}$
will be the product measure of $\nu_2$ on $\C^M$.  By combining
$\nu_L$ and $\nu_{2M}$ with a certain Vandermonde determinant we will
arrive at the desired measures on $\R^L \times \CnoR^M$.  Given a
vector $\bs{\gamma} \in \C^N$ we define $V(\bs{\gamma})$ to be the $N
\times N$ matrix whose $n,n'$ entry is given by $\gamma_{n}^{n'-1}$.  We will denote the determinant of $V(\bs{\gamma})$ by
$\Delta(\bs{\gamma})$, and define the function $\Delta :
\R^L \times \C^M \rightarrow \C$ by
\[
\Delta(\bs{\alpha}, \bs{\beta}) := \Delta(\alpha_1, \ldots,
\alpha_L, \beta_1, \overline{\beta_1}, \ldots, \beta_M,
\overline{\beta_M}).  
\]

Using these definitions we set $\nu_{L,M}$ to be the measure on $\R^L
\times \CnoR^M$ given by
\begin{equation}
\label{eq:8}
\frac{d\nu_{L,M}}{d(\nu_L \times \nu_{2M}) }(\bs{\alpha}, \bs{\beta}) =  2^M |\Delta(\bs{\alpha}, \bs{\beta})|,
\end{equation}
and we will write $\mathcal{P}_{L,M}^{\nu}$ for the measure on
$X_{L,M}$ given specified by $\nu_{L,M}$ as in (\ref{eq:3}).  If
$\mathcal{P}^{\nu}_{L,M}$ is finite for each pair $(L,M)$, then we 
set 
\[
\mathcal{P}^{\nu} := \frac{1}{Z^{\nu}} \sum_{(L,M)} 
\mathcal{P}^{\nu}_{L,M} \qquad \mbox{where} \qquad Z^{\nu} :=
\sum_{(L,M)} \mathcal{P}_{L,M}^\nu(X_{L,M}).
\]
We will call $\mathcal{P}^{\nu}$ the point process on $X_N$ associated
to the {\em weight measure} $\nu$ and $Z^{\nu}$ will be referred to as
the {\em partition function} of $\mathcal{P}^{\nu}$.  

We set $\lambda_1$ and $\lambda_2$ to be Lebesgue measure on $\R$ and
$\C$ respectively and let $\lambda := \lambda_1 + \lambda_2$.  If
there exists a function $w: \C \rightarrow[0,\infty)$ such that $\nu =
w\lambda$ (by which we mean $d\nu/d\lambda = w$), then we will call
$w$ the {\em weight function} of $\mathcal{P}^{\nu}$.  In this situation we
set $\mathcal{P}_{L,M}^{w} := \mathcal{P}_{L,M}^{\nu}$ and define
$\mathcal{P}^w$ and $Z^w$ analogously.  If $P^{\nu}$ has weight function
$w$ then
\begin{equation}
\label{eq:14}
\frac{d
  \nu_{L,M}}{d(\lambda_L \times \lambda_{2M})}(\bs{\alpha},
\bs{\beta}) = 2^M
\bigg\{ \prod_{\ell=1}^L w(\alpha_{\ell}) \prod_{m=1}^M
w(\beta_m) \bigg\} \big| \Delta(\bs{\alpha},\bs{\beta}) \big|,
\end{equation}
and we define $\Omega_{L,M}: \R^L \times \C^M$ to be the function
given on the right hand side of (\ref{eq:14}).  The collection
$\{\Omega_{L,M} : L + 2M=N\}$ plays the role of the 
joint eigenvalue probability density function, and we will call
$\Omega_{L,M}$ the $L,M$-{\em partial} joint eigenvalue probability
density function of $\mathcal{P}^w$.  

\section{Examples of Point Processes Associated to Weights}

Point processes on $X_N$ associated to weights arise in a variety of
contexts. It is often the case that the weight function $w$ is
invariant under complex conjugation.  In this situation, there
necessarily exists some function $\rho: \C \rightarrow [0,\infty)$
such that
\[
w(\gamma) := \left\{
\begin{array}{ll}
\rho(\gamma) & \quad \mbox{if } \gamma \in \R; \\ & \\
\rho(\gamma) \rho(\overline{\gamma}) & \quad \mbox{if } \gamma \in
\CnoR. 
\end{array}
\right.
\]
When it exists, we will cal $\rho$ the {\em root} or {\em eigenvalue}
function of $\mathcal{P}^w$ (depending on whether $\mathcal{P}^w$
models the roots of random polynomials or eigenvalues of random
matrices).  The root/eigenvalue function is often a more natural
descriptor of $\mathcal{P}^w$ than $w$.  In this situation we will
write $\mathcal{P}^{[\rho]}$ and $Z^{[\rho]}$ for $\mathcal{P}^w$ and
$Z^w$.  

\subsection{Ginibre's Real Ensemble}

In 1965 J.~Ginibre introduced three ensembles of random matrices whose
entries were respectively chosen with Gaussian density from $\R$, $\C$
and Hamilton's quaternions \cite{MR0173726}.  Ginibre's real ensemble
is given by $\R^{N \times N}$ together with the probability measure 
$\eta$ given by
\[
d\eta(\mathbf{Y}) = (2 \pi)^{-N^2/2} e^{-\Tr(\mathbf{Y}^{\mathrm{T}} \mathbf{Y})/2} \, d\lambda_{N
  \times N}(\mathbf{Y}), 
\]
where $\lambda_{N \times N}$ is Lebesgue measure on $\R^{N \times N}$.
This ensemble has since been named GinOE due to certain similarities
with the Gaussian Orthogonal Ensemble.  Among Ginibre's original goals
was to produce a formula for the partial joint eigenvalue probability
density functions of GinOE.  He was only able to do this for the
subset of matrices with all real eigenvalues.  In the 1990s Lehmann
and Sommers \cite{1991PhRvL..67..941L} and later Edelman
\cite{MR1437734} proved that the $L,M$-partial eigenvalue probability
density function of GinOE is given by $\Omega_{L,M}$ for the
eigenvalue function 
\[
\mathrm{Gin}(\gamma) := \exp(-\gamma^2/2) \sqrt{\erfc(\sqrt{2} |\ip{\gamma}|)}.
\]
Consequently, the investigation of the eigenvalue statistics of GinOE
reduces to the study of $\mathcal{P}^{[\mathrm{Gin}]}$.  

\subsection{The Range of Mahler measure}

The {\em Mahler measure} of a polynomial 
\begin{equation}
\label{eq:2}
f(x) = \sum_{n=0}^N a_n x^{N-n} = a_0 \prod_{n=1}^N (x - \gamma_n)
\end{equation}
is given by
\[
\mu(f) = \exp\left\{
\int_0^1 \log\big| f(e^{2 \pi i \theta}) \big| \, d\theta
\right\} = |a_0| \prod_{n=1}^N \max\{1, |\gamma_n| \}.
\]
The second equality comes from Jensen's formula.  Mahler measure
arises in a number of contexts, {\it i.e.} ergodic theory, potential
theory and Diophantine geometry and approximation (a good reference
covering the many aspects of Mahler measure is \cite{MR1700272}).  One
problem in the context of the geometry of numbers is to estimate the
number of degree $N$ integer polynomials with Mahler measure bounded
by $T > 0$ as $T \rightarrow \infty$.  Chern and Vaaler produced such
an estimate in \cite{MR1868596} using the general principal that the
number of lattice points in a `reasonable' domain in $\R^{N+1}$ is
roughly equal to the volume (Lebesgue measure) of the domain.  That
is, if we identify degree $N$ polynomials with their vector of
coefficients in $\R^{N+1}$ and use the approximation
\begin{align*}
\#\{ g(x) \in \Z[x] : \deg g = N, \mu(g) \leq T \} &\approx
\lambda_{N+1}\{ g(x) \in \R[x] : \deg g =N, \mu(g) \leq T \} \\
&= T^{N+1} \lambda_{N+1}\{ g(x) \in \R[x] : \deg g =N, \mu(g) \leq 1 \},
\end{align*}
then the main term in the asymptotic estimate Chern and Vaaler were
interested in can be expressed in terms of the volume of the degree
$N$ {\em star body} of Mahler measure,
\begin{equation}
\label{eq:5}
\mathcal{U}_N = \{ g(x) \in \R[x] : \deg g = N, \mu(g) \leq 1 \}.
\end{equation}
The volume of this set is given by
\begin{equation}
\label{eq:6}
\frac{2}{N+1} Z^{[\psi]}
\end{equation}
where $\psi(\gamma) = \max\{1, |\gamma|\}^{-N-1}$ \cite{MR1868596}.
Consequently, the volume of the star body which leads to an asymptotic
estimate of interest in Diophantine geometry is essentially equals
the partition function for the random point process $\mathcal{P}^{[\psi]}$.

\subsection{The Range of Other Multiplicative Functions on
  Polynomials}
We can generalize Mahler measure by replacing the function $\gamma
\mapsto \max\{1, |\gamma|\}$ with other functions of $\gamma$.  Given
a continuous function $\phi: \C \rightarrow (0,\infty)$ which 
satisfies the asymptotic formula,
\begin{equation}
\label{eq:22}
\phi(\gamma) \sim |\gamma| \qquad \mbox{as} \qquad |\gamma|
\rightarrow \infty,
\end{equation}
we define the function $\Phi: \C[x] \rightarrow [0,\infty)$ by
\[
\Phi: a_0 \prod_{n=1}^N (x - \gamma_n) \mapsto |a_0| \prod_{n=1}^N
\phi(\gamma_n).
\]
The function $\Phi$ is known as a {\em multiplicative distance
  function} (so named because it is a distance functions in the
sense of the geometry of numbers on finite dimensional subspaces of
$\C[x]$).  The asymptotic condition in (\ref{eq:22}) ensures that
$\Phi$ is continuous on finite dimensional subspaces of $\C[x]$ (one
of the axioms of a distance functions). 

We define the degree $N$ starbody of $\Phi$ in analogy with 
(\ref{eq:5}):
\[
\mathcal{U}_N(\Phi) = \{ g(x) \in \R[x] : \deg g = N, \Phi(g) \leq 1 \}.
\]
In this situation the volume of $\mathcal{U}_N(\Phi)$ equals
\[
\lambda_{N+1}(\mathcal{U}_N(\Phi)) = \frac{2}{N+1} Z^{[\psi]}, \qquad
\mbox{where} \qquad \psi(\gamma) = \phi(\gamma)^{-N-1}.
\]

We may discover Further information about the range of values $\Phi$
takes on degree $N$ real polynomials by considering the point process
on $X_N$ corresponding to the root function $\psi(\gamma) =
\phi(\gamma)^{-\sigma}$ 
for $\sigma > N$.  The partition function of $\mathcal{P}^{[\psi]}$ is therefore a
function of $\sigma$, and $Z^{[\psi]}(\sigma)$ is known as the degree
$N$ {\em moment function} of $\Phi$.  In fact, we may extend the
domain of a real moment functions to a function of a complex variable $s$ on the
half-plane $\rp{s} > N$.  Moreover, in this domain $Z^{[\psi]}(s)$ is
analytic. Any analytic continuation beyond this half-plane gives
information about the range of values of $\Phi$ which may not be
realizable by other methods.  For instance, when $\psi(\gamma) =
\max\{1, |\gamma|\}^{-s}$, Chern and Vaaler discovered that
$Z^{[\psi]}(s)$ has an analytic continuation to a rational function of
$s$ with rational coefficients and poles at positive integers $\leq N$
\cite{MR1868596}.  Similar results have been found for moment
functions of other multiplicative distance functions; see
\cite{sinclair-2005}.  

\section{Correlation Measures and Functions}
Suppose $\ell$ and $m$ are non-negative integers, not both equal to
0.  
Then, given a function  $f: \R^{\ell} \times \CnoR^m \rightarrow \C$
we define the function 
$F_f: X_N \rightarrow \C$ by
\[
F_f(\xi) = \sum_{\{\mathbf{x}, \mathbf{z}\} \subseteq \xi}
f(\mathbf{x}, \mathbf{z}),
\]
where the sum is over all $(\mathbf{x}, \mathbf{z}) \in \R^{\ell}
\times \CnoR^m$ such that $\{  \mathbf{x}, \mathbf{z} \} \subseteq \xi$.
We take an empty sum to equal 0, and thus if $\xi \in X_{L,M}$ with $L
\leq \ell$ or $M \leq m$, then $F_f(\xi) = 0$.  

Given a point process $\mathcal{P}$ on $X_N$, if there exists a
measure $\rho_{\ell, m}$ on $\R^{\ell} \times \CnoR^m$ 
such that for all Borel measurable functions $f$,
\begin{equation}
\label{eq:7}
\int_{\R^{\ell} \times {\C_{\ast}^m}} f(\mathbf{x}, \mathbf{z}) \,
d\rho_{\ell,m}(\mathbf{x}, \mathbf{z}) = \int_{X_N} F_f(\xi) \, d\mathcal{P}(\xi),
\end{equation}
then we call $\rho_{\ell,m}$ the $(\ell, m)$--{\em correlation measure} of
$\mathcal{P}$.  Furthermore, if $\rho_{\ell,m}$ has a density with respect to $\lambda_{\ell} \times
\lambda_{2m}$ then we will call this density the $\ell,m$--{\em
  correlation function} of $\mathcal{P}$ and denote it by
$R_{\ell,m}$.  By convention we will take $R_{0,0}$ to be the contant
1. 
\begin{prop*}
\label{prop:1} ${\displaystyle \frac{d\rho_{\ell, m}}{d(\nu_{\ell} \times
  \nu_{2m})}(\mathbf{x}, \mathbf{z})}$ equals
\[\frac{2^m}{Z^{\nu}} \sum_{(L,M) \atop  L \geq  \ell,  M \geq m}
  \frac{1}{(L-\ell)! (M-m)!} \int\limits_{\R^{L-\ell}}  
\int\limits_{\C^{M-m}}   |\Delta(\mathbf{x} \! \vee \!
\bs{\alpha}, \mathbf{z} \! \vee \! \bs{\beta})| \,
d\nu_{L-\ell}(\bs{\alpha}) \, d\nu_{2(M-m)}(\bs{\beta}),
\]
where $\mathbf{x} \vee \bs{\alpha} \in \R^L$ is the vector formed by
concatenating the vectors $\mathbf{x} \in \R^{\ell}$ and $\bs{\alpha}
\in \R^{L-\ell}$ (and similarly for $\mathbf{z} \! \vee \! \bs{\beta} \in \C^M$).
\end{prop*}

\begin{cor*} 
If $\nu = w\lambda$, then $R_{\ell, m}(\mathbf{x}, \mathbf{z})$ equals
\[
\frac{1}{Z^{w}} \sum_{(L,M)
    \atop  L \geq \ell, M \geq
    m} \frac{1}{(L-\ell)! (M-m)! 2^{M-m}} 
\int\limits_{\R^{L-\ell}} \int\limits_{\C^{M-m}}
\Omega_{L,M}(\mathbf{x} \! \vee \! \bs{\alpha}, \mathbf{z} \! \vee \!
\bs{\beta}) \,  d\lambda_{L-\ell}(\bs{\alpha}) \,
d\lambda_{2(M-m)}(\bs{\beta}). 
\]
\end{cor*}

\begin{proof}[Proof of Proposition~\ref{prop:1}]
Assume that $L \geq \ell$ and $M \geq m$, from (\ref{eq:3}) and (\ref{eq:8}),  
\begin{equation}
\label{eq:11}
\int_{X_{L,M}} F_f(\xi) \, d\mathcal{P}^{\nu}_{L,M}(\xi) = \frac{1}{L! M!}
\int_{\R^L} \int_{\CnoRsub^M} F_f\{\bs{\alpha}, \bs{\beta}\} \, 
|\Delta(\bs{\alpha}, \bs{\beta})| \, d\nu_L(\bs{\alpha}) \,
d\nu_{2M}(\bs{\beta}),
\end{equation}
The function $(\bs{\alpha}, \bs{\beta}) \mapsto | \Delta(\bs{\alpha},
\bs{\beta}) |$ is invariant under any permutation of the coordinates
of $\bs{\alpha}$ and $\bs{\beta}$, since such a permutation merely
permutes the columns of the Vandermonde matrix.  Similarly, replacing
any of the $\beta$ with their complex conjugates merely transposes
pairs of columns of the Vandermonde matrix.  That is, if $(\bs{\alpha},
\bs{\beta})$ and $(\bs{\alpha'}, \bs{\beta}')$ are elements in $\R^L
\times \CnoR^M$ such that $\{\bs{\alpha}, \bs{\beta}\} =
\{\bs{\alpha}', \bs{\beta}'\}$, then $| \Delta(\bs{\alpha},
\bs{\beta}) | = | \Delta(\bs{\alpha}', \bs{\beta}') |$.  Moreover, if
any of the $\beta$ are real, then the Vandermonde matrix has two 
identical columns and is therefore zero.  We may thus replace the
domain of integration on the right hand side of (\ref{eq:11}) with
$\R^L \times \C^M$. In fact, we may assume that the domain of
integration on the right hand side is over the subset of $\R^L \times
\C^M$ consisting of those vectors with distinct coordinates.  

Now,
\[
F_f\{\bs{\alpha}, \bs{\beta}\} =\sum_{\{\mathbf{x}, \mathbf{z}\}
  \subseteq \{ \bs{\alpha},  \bs{\beta}\}} f(\mathbf{x}, \mathbf{z}).
\]
Assuming that the coordinates of $(\bs{\alpha}, \bs{\beta})$ are
distinct, and if $(\mathbf{x}, \mathbf{z}) \in \R^{\ell} \times \C^m$
is such that $\{ \mathbf{x}, \mathbf{z} \} \subseteq \{ \bs{\alpha},
\bs{\beta} \}$ we may find a vector
$(\mathbf{a}, \mathbf{b}) \in \R^{L - \ell} \times \C^{M - m}$ such
that $(\mathbf{x} \vee \mathbf{a}, \mathbf{z} \vee \mathbf{b})$ is given
by permuting the coordinates of $(\bs{\alpha}, \bs{\beta})$.  Clearly
$|\Delta(\bs{\alpha}, \bs{\beta})| =
|\Delta(\mathbf{x} \vee \mathbf{a}, \mathbf{z} \vee \mathbf{b})|$
and
\[
d\nu_L(\bs{\alpha}) \, d\nu_{2M}(\bs{\beta}) =
d\nu_{\ell}(\mathbf{x}) \, d\nu_{L-\ell}(\mathbf{a}) \,
d\nu_{m}(\mathbf{z}) \, d\nu_{M-m}(\mathbf{b}).
\]
These observations together with an application of Fubini's
Theorem imply that (\ref{eq:11}) can be written as 
\[
\frac{1}{L! M!} \int_{\R^{\ell}} \int_{\C^m} \sum_{\{\mathbf{x},
  \mathbf{z}\} } f(\mathbf{x}, \mathbf{z}) \Bigg\{
\int\limits_{\R^{L-\ell}} \int\limits_{\C^{M-m}}
|\Delta(\mathbf{x}\!\vee\!\mathbf{a}, \mathbf{z}\!\vee\!\mathbf{b})| \,
d\nu_{L-\ell}(\mathbf{a}) \, d\nu_{M-m}(\mathbf{b}) \Bigg\} \,
d\nu_{\ell}(\mathbf{x}) \, d\nu_{m}(\mathbf{z}).
\]
Now, it is easily seen that there are 
\[
2^m {L \choose \ell} \ell! {M \choose m} m!
\]
 vectors corresponding to each $\{\mathbf{x}, \mathbf{z}\}$, and thus we find 
\begin{align*}
&  \int_{X_{L,M}} F_f(\xi) \, d\mathcal{P}^{\nu}_{L,M}(\xi) =
  \int_{\R^{\ell}} \int_{\C^m} f(\mathbf{x}, \mathbf{z}) \\ & \qquad
  \times \Bigg\{ \frac{2^m}{(L-\ell)!(M-m)!}
  \int\limits_{\R^{L-\ell}} \int\limits_{\C^{M-m}}
  |\Delta(\mathbf{x}\!\vee\!\mathbf{a}, \mathbf{z}\!\vee\!\mathbf{b})| \,
  d\nu_{L-\ell}(\mathbf{a}) \, d\nu_{M-m}(\mathbf{b}) \Bigg\} \,
  d\nu_{\ell}(\mathbf{x}) \, d\nu_{m}(\mathbf{z}).
\end{align*}
The proposition follows since 
\[
\int_{X} F_f(\xi) \, d\mathcal{P}^{\nu}(\xi) = \frac{1}{Z^{\nu}} \sum_{(L,M)}
\int_{X_{L,M}} F_f(\xi) \, d\mathcal{P}^{\nu}_{L,M}(\xi). \qedhere
\]
\end{proof}

From here forward, and unless otherwise stated, $(\ell, m)$ will
represent an ordered pair of non-negative integers such that $\ell +
2m \leq N$.  

\section{A Matrix Kernel for Point Processes on $X_N$}
From here forward we will assume that $N$ is even.  

Let $\mathcal{P}^{\nu}$ be the point process on $X_N$ associated to the
weight $w$, and as before let $\nu = \nu_1 + \nu_2$.  We
define the operator $\epsilon_{\nu}$ on the Hilbert space $L^2(\nu)$ by 
\[
\epsilon_{\nu} g(\gamma) := \left\{
  \begin{array}{ll} {\displaystyle \frac{1}{2}\int_{\R} g(y) \sgn(y -
      \gamma) \, d\nu_1(y)} & \quad \mbox{if} \quad
    \gamma \in \R, \\ & \\
    {\displaystyle i g(\overline{\gamma}) \sgn(\ip{\gamma})} & \quad
    \mbox{if} \quad \gamma \in \CnoR,
\end{array}
\right.
\]
and we use this to define the skew-symmetric bilinear form on
$L^2(\nu)$ given by $\la \cdot | \cdot \ra_{\nu}$
\[
\la g | h \ra_{\nu} := \int_{\C} g(\gamma) \epsilon_{\nu} h(\gamma) -
\epsilon_{\nu} g(\gamma) h(\gamma) \, d\nu(\gamma).
\]
If $\nu = w\lambda$ then $\la g | h
\ra_{\nu} = \la \widetilde{g} | \widetilde{h} \ra_{\lambda}$ where,
for instance, $\widetilde{h}(\gamma) := h(\gamma)w(\gamma)$.  
\begin{thm*}
\label{thm:1}
Let $\mathbf{q} = \{ q_0(\gamma), q_2(\gamma), \ldots,
q_{N-1}(\gamma) \} \subseteq \R[\gamma]$ be such that each $q_n$ is
monic and $\deg q_n = n$.  Then,
\begin{equation}
\label{eq:4}
Z^{\nu} = \Pf \mathbf{U}^{\nu}_{\mathbf{q}},
\end{equation}
the Pfaffian of $\mathbf{U}^{\nu}_{\mathbf{q}}$, where 
\[
\mathbf{U}^{\nu}_{\mathbf{q}} := [\la q_n | q_{n'} \ra_{\nu}]; \qquad n,n'=0,1,\ldots,N-1.
\]
\end{thm*}
We will call $\mathbf{q}$ a {\em complete family} of monic
polynomials.
\begin{rem}
It is at this point that it is necessary that $N$ be even, since the
Pfaffian is only defined for antisymmetric matrices with an even
number of rows and columns.  A similar formula to (\ref{eq:4}) exists
for $Z^{\nu}$ in the case when $N$ is odd; see \cite{sinclair-2007}.
However, we have not pursued the subsequent 
analysis necessary to recover the correlation functions in this case. 
\end{rem}

Theorem~\ref{thm:1} follows from results proved in
\cite{sinclair-2007}.  In fact, \cite{sinclair-2007} gives a formula
for the average of a multiplicative class function over the point
process on $X_N$ determined by Ginibre's weight.  However, the
combinatorics necessary to arrive at such averages is independent of
any specific feature of ${\mathrm{Gin}}(\gamma)$ and the measure on
$\C$ specified by ${\mathrm{Gin}(\gamma)} d\lambda(\gamma)$ can
formally be replaced by any measure $\nu$.

In order to express the correlation functions for the point process
$\mathcal{P}^w$ associated to the weight $w$ we will define
$\eta$ to be a measure on $\C$ given by a linear combination of point
masses, and then use the definition of the partition function and
properties of Pfaffians to expand both sides of the equation
\begin{equation}
\label{eq:9}
\frac{Z^{w (\lambda + \eta)}}{Z^w} = \frac{1}{Z^w} \Pf
\mathbf{U}_{\mathbf{q}}^{w(\lambda + \eta)}  
\end{equation}
The coefficients in the linear combination defining $\eta$ appear
again in terms on both sides of the expanded equation, and after
identifying like coefficients on both sides of the expanded equation
we will be able to read off a closed form for the correlation
functions in terms of the Pfaffian of a matrix whose entries depend on
a $2 \times 2$ matrix kernel.  

In order to define the $2 \times 2$ matrix kernel for $\mathcal{P}^w$,
we let $\mathbf{q}$ be a complete family of monic polynomials, and we 
define 
\begin{equation}
\label{eq:10}
\widetilde{q}_n(\gamma) := q_n(\gamma) w(\gamma) \qquad n=0,1,\ldots,N-1.
\end{equation}
We then define
\[
S_N(\gamma, \gamma') := 2 \sum_{n=0}^{N-1} \sum_{n'=0}^{N-1}
\mu_{n,n'} \, \widetilde{q}_n(\gamma) \, \epsilon_{\lambda}
\widetilde{q}_{n'}(\gamma'),
\]
where we define $\mu_{n,n'}$ to be the $n,n'$ entry of
$(\mathbf{U}^w_{\mathbf{q}})^{-\mathsf{T}}$. 
Similarly we define,
\[
IS_N (\gamma, \gamma') := 2 \sum_{n=0}^{N-1} \sum_{n'=0}^{N-1}
\mu_{n,n'} \, \epsilon_{\lambda} \widetilde{q}_n(\gamma) \, 
\epsilon_{\lambda} \widetilde{q}_{n'}(\gamma')  
\]
and
\[
DS_N (\gamma, \gamma') := 2 \sum_{n=0}^{N-1} \sum_{n'=0}^{N-1}
\mu_{n,n'} \, \widetilde{q}_n(\gamma) \, \widetilde{q}_{n'}(\gamma').
\]
\begin{rem}
  The functions $S_N, IS_N$ and $DS_N$ can be shown to be independent
  of the family $\mathbf{q}$.  By setting $\mathbf{q}$ to be
  skew-orthogonal with respect to the bilinear form $\la \cdot | \cdot
  \ra_{\lambda}$ we arrive at particularly simple representations for
  these expressions.
\end{rem}

Finally, in order to define the matrix kernel $\mathcal{P}^w$ we
define the function \mbox{$\mathcal{E}: \C^2  \rightarrow
  \{-\frac{1}{2}, 0, \frac{1}{2}\}$} by 
\[
\mathcal{E}(\gamma, \gamma') := \left\{
\begin{array}{ll}
{\displaystyle \frac{1}{2} \sgn(\gamma - \gamma')} & \quad \mbox{if }
  \gamma, \gamma' \in \R; \\ & \\
0 & \quad \mbox{otherwise.}
\end{array}
\right.
\]
The matrix kernel of $\mathcal{P}^w$ is then given by
\begin{equation}
\label{eq:19}
K_N(\gamma, \gamma') := \begin{bmatrix}
DS_N(\gamma, \gamma') & S_N(\gamma, \gamma') \\
-S_N(\gamma', \gamma) & IS_N(\gamma, \gamma') + \mathcal{E}(\gamma,
\gamma')
\end{bmatrix}.
\end{equation}
\begin{rem}
The explicit $N$-dependence of $K_N$ and its constituents is
traditional, since one is often interested in the $N \rightarrow
\infty$ asymptotics of $K_N$.
\end{rem}

\section{Correlation Functions in Terms of the Matrix Kernel}

Finally we may state the main result of this manuscript.  
\begin{thm*}
\label{thm:2}

The $\ell,m$--correlation function of $\mathcal{P}^w$ is given by
\[
R_{\ell,m}(\mathbf{x}, \mathbf{z}) = \Pf \begin{bmatrix}
K_N(x_j, x_{j'}) & K_N(x_j, z_{k'} ) \\
K_N(z_k , x_{j'}) & K_N(z_k , z_{k'})
\end{bmatrix}; \qquad 
\begin{array}{ll}
j,j'=1,2,\ldots,\ell; \\
k,k'=1,2,\ldots,m.
\end{array}
\]
\end{thm*}
\begin{rem}
The matrix on the right hand side of this expression is composed of $2
\times 2$ blocks, so that, for instance, the first row of $2 \times 2$
blocks is given by
\[
\big[
K_N(x_1, x_1) \quad \ldots \quad K_N(x_1, x_{\ell}) \quad K_N(x_1, z_1)
\quad \ldots \quad K_N(x_1, z_m) 
\big].
\]
\end{rem}

We define the measure $\delta$ on $\C$ to be the measure with unit
point mass at 0.  Given $U$ real numbers $x_1, x_2, \ldots, x_U$ and
$V$ non-real complex numbers, $z_1, z_2, \ldots, z_V$, we define the
measure $\eta$ by
\[
d \eta(\gamma) := \sum_{u=1}^U a_u \, d\delta(\gamma - x_u) + \sum_{v=1}^V b_v
\,\big( d \delta(\gamma - z_v) + d\delta(\gamma - \overline{z_v}) \big),
\]
where $a_1, a_2, \ldots, a_U$ and $b_1, b_2, \ldots, b_V$ are
indeterminants.  It does no harm to assume that $U$ and $V$ are both
greater than $N$.  By reordering and renaming the $x, z, a$ and $b$ we
will also write
\[
d \eta(\gamma) := \sum_{t=1}^T c_t \, d \widehat{\delta}(\gamma - y_t),
\]
where we define
\[
d \widehat{\delta}(\gamma - y) = 
\left\{
\begin{array}{ll}
d\delta(\gamma - y) & \quad \mbox{if } y \in \R; \\
& \\
 d \delta(\gamma - y) + d\delta(\gamma - \overline{y}) & \quad
 \mbox{if } y \in \CnoR.  
\end{array}
\right.
\]
Clearly $T = U + V$.

As we alluded to previously, the proof of Theorem~\ref{thm:2} relies
on expanding  $Z^{w(\lambda+\eta)}/Z^w$ in two different ways
and then equating the coefficients of certain products of $c_1, c_2,
\ldots, c_T$.  One of the expansions of
$Z^{w(\lambda+\eta)}/Z^w$ comes from Theorem~\ref{thm:1},
while the other comes directly from the definition of the partition
function. 

\begin{lemma*}
\label{lemma:4}
\begin{equation}
\label{eq:23}
\frac{Z^{w(\lambda + \eta)}}{Z^w} = \Pf \big( \mathbf{J} + \big[ \sqrt{c_t
  c_{t'}} K_N(y_t, y_{t'}) \big] \big); \qquad t,t'=1,2,\ldots,T,
\end{equation}
where $\mathbf{J}$ is defined to be the $2T \times 2T$ matrix
consisting of $2 \times 2$ blocks given by
\[
\mathbf{J} := \left[
\delta_{t,t'} \begin{bmatrix}
0 & 1 \\
-1 & 0 
\end{bmatrix}
\right];
\qquad t,t = 1,2,\ldots, T.
\]
\end{lemma*}
\begin{rem}
The Pfaffian which appears in (\ref{eq:23}) is an example of a {\em
  Fredholm Pfaffian}.  This is the Pfaffian formulation of the notion
of a Fredholm determinant and is discussed in \cite{rains-2000}.
\end{rem}

We defer the proof of Lemma~\ref{lemma:4} and Lemma~\ref{lemma:5}
until after the proof of Theorem~\ref{thm:2}.  

For each $(\ell, m)$ and $(L,M)$ we define the $\ell,m,L,M$-{\em
  partial} correlation function of $\mathcal{P}^w$ to be
$R_{\ell,m,L,M}: \R^{\ell} \times \C^m \rightarrow [0,\infty)$ where
\begin{align*}
&R_{\ell,m,L,M}(\mathbf{x}, \mathbf{z}) := \\
& \qquad \frac{1}{(L-\ell)! (M-m)! 2^{M-m}}
\int\limits_{\R^{L-\ell}}\int\limits_{\C^{M-m}}
\Omega_{L,M}(\mathbf{x} \vee \bs{\alpha}, \mathbf{z} \vee \bs{\beta}) \,
d\lambda_{L-\ell}(\bs{\alpha}) \, d\lambda_{2(M-m)}(\bs{\beta}).
\end{align*}
When $\ell = m = 0$, we take $\mathbf{x} \vee \bs{\alpha} =
\bs{\alpha}$ and $\mathbf{z} \vee \bs{\beta} = \bs{\beta}$, so that
$R_{0,0,L,M}$ is a constant equal to $\mathcal{P}_{L,M}(X_{L,M})$.
The partial correlation functions are related to the correlations
functions by the formula 
\begin{equation}
\label{eq:24}
R_{\ell,m}(\mathbf{x}, \mathbf{z}) = \frac{1}{Z^w} \sum_{(L,M) \atop L
  \geq \ell, M \geq m} R_{\ell,m,L,M}(\mathbf{x}, \mathbf{z}),
\end{equation}
and thus the partial correlation functions are one path to the
correlation functions.  Equation~(\ref{eq:24}) is still valid when
$\ell = m =0$, though the constituent correlation `functions' are
actually constants.

The partial correlation functions of $\mathcal{P}^{[\mathrm{Gin}]}$ of
the special forms $R_{L,M,0,M}$ and $R_{0,N,0,n}$ have been studied by
Akemann and Kanzieper in \cite{akemann-2007} and
\cite{kanzieper-2005-95}.  For general weight function $w$, the
partial correlation functions of the form $R_{N,0,n,0}$ are (up to
normalization) equal to the correlation functions of the $\beta=1$
Hermitian ensemble with weight $w$.  This connection will be exploited
in \ref{app:A}.  

Before stating the next lemma we need a bit of notation.  Given
non-negative integers $n$ and $W$, we define $\mf{I}_n^W$ to be the
set of increasing functions from $\{1,2,\ldots,n\}$ into
$\{1,2,\ldots,W\}$.  Clearly if $W < n$ then $\mf{I}_n^W$ is empty.
Given a vector $\mathbf{a} \in \C^W$ and an element $\mf{u} \in
\mf{I}_n^W$, we define the vector $\mathbf{a}_{\mf{u}} \in \C^n$ by
$\mathbf{a}_{\mf{u}} = \{a_{\mf{u}(1)}, a_{\mf{u}(2)}, \ldots,
a_{\mf{u}(n)}\}$.  
\begin{lemma*}
\label{lemma:5}
For each pair $(L,M)$, 
\begin{align*}
&\frac{1}{L! M! 2^M} \int_{\R^L} \int_{\C^M} \Omega_{L,M}(\bs{\alpha},
\bs{\beta}) \, d(\lambda + \eta)_L(\bs{\alpha}) \, d(\lambda +
\eta)_{2M}(\bs{\beta}) = \\
& \hspace{3cm} \sum_{\ell=0}^L \sum_{m=0}^M \sum_{\mf{u} \in \mf{I}_{\ell}^U}
\sum_{\mf{v} \in \mf{I}_m^V} \bigg\{
\prod_{j=1}^{\ell} a_{\mf{u}(j)} \prod_{k=1}^m b_{\mf{v}(k)} 
\bigg\} R_{\ell,m,L,M}(\mathbf{x}_{\mf{u}}, \mathbf{z}_{\mf{v}}).
\end{align*}
\end{lemma*}
\begin{rem}
We will use the convention that 
\[
\sum_{\mf{u} \in \mf{I}_0^U} \prod_{j=1}^0 a_{\mf{u}(j)} = \sum_{\mf{v}
  \in \mf{I}_0^V} \prod_{k=1}^0 b_{\mf{v}(k)} = 1.
\]
This will allow us to keep from having to deal with the pathological
correlation `functions' $R_{0,0,L,M}$ and $R_{0,0}$ separately.  
\end{rem}
\begin{lemma*}
\label{lemma:6}
Suppose $K : \{1,2,\ldots,T\} \times \{1,2,\ldots,T\} \rightarrow \R^{2
\times 2}$ is such that
\[
K(t,t') = -K(t',t)^{\transpose},
\]
and define $\mathbf{K}$ to be the $2T \times 2T$ block antisymmetric
matrix whose $t,t'$ entry is given by $K(t,t')$.  Then,
\begin{equation}
\label{eq:25}
\Pf[ \mathbf{J} + \mathbf{K} ] = 1 + \sum_{S=1}^T \sum_{\mf{t} \in
  \mf{I}_S^T} \Pf \mathbf{K}_{\mf{t}},
\end{equation}
where for each $\mf{t} \in \mf{I}_S^T$, $\mathbf{K}_{\mf{t}}$ is the
$2S \times 2S$ antisymmetric matrix given by
\[
\mathbf{K}_{\mf{t}} = [ K(\mf{t}(s), \mf{t}(s'))]; \qquad s,s'=1,2,\ldots,S.
\]

\end{lemma*}
\begin{proof}
This is a special case of the formula for the Pfaffian of the sum of
two antisymmetric matrices.  See \cite{sinclair-2005} or
\cite{MR1069389} for a proof.
\end{proof}

Using these lemmas we may complete the proof of Theorem~\ref{thm:2}.
First notice that Lemma~\ref{lemma:5} and (\ref{eq:24}) imply that
\begin{equation}
\label{eq:26}
\frac{Z^{w(\lambda + \eta)}}{Z^w} = \sum_{(l,m)} \sum_{\mf{u} \in
  \mf{I}_{\ell}^U} \sum_{\mf{v} \in \mf{I}_{m}^V} \bigg\{
\prod_{j=1}^{\ell} a_{\mf{u}(j)} \prod_{k=1}^m b_{\mf{v}(k)}
\bigg\} R_{\ell, m}(\mathbf{x}_{\mf{u}}, \mathbf{z}_{\mf{v}}).
\end{equation}
This follows by summing both sides of the expression in
Lemma~\ref{lemma:6} over all $(L,M)$ and then reorganizing the sums
over $(L,M)$, $\ell$ and $m$.  

Given $\mf{u} \in \mf{I}_{\ell}^U$ and $\mf{v} \in \mf{I}_m^V$ we
define $(\mf{u}\!:\!\mf{v})$ to be the element in
$\mf{I}_{\ell+m}^{U+V}$ given by
\[
(\mf{u} \vee \mf{v})(n) := \left\{
\begin{array}{ll}
\mf{u}(n) & \quad \mbox{if } n \leq \ell \\
\ell + \mf{v}(n) & \quad \mbox{if } n > \ell.
\end{array}
\right.
\]
Notice that each $\mf{t} \in \mf{I}_S^T$ is equal to $(\mf{u} \vee \mf{v})$ for some $\mf{u} \in \mf{I}_{\ell}^U$ and
$\mf{v} \in \mf{I}_m^V$ with $\ell + m = S$.  (This does not preclude
the possibility that either $U$ or $V$ equals 0, in which case $\mf{t} =
\mf{v}$ or $\mf{t} = \mf{u}$).  It follows that we can rewrite
(\ref{eq:25}) as
\[
\Pf(\mathbf{J} + \mathbf{K}) = \sum_{(\ell,m) \atop \ell + m \leq T}
\sum_{\mf{u} \in   \mf{I}_{\ell}^U} \sum_{\mf{v} \in \mf{I}_{m}^V} \Pf
\mathbf{K}_{(\mf{u}\vee\mf{v})}.
\]
If we set $K_{t,t'}(t,t') = \sqrt{c_t  c_{t'}} \, K_N(y_t, y_{t'})$,
then 
\[
\Pf
\mathbf{K}_{(\mf{u}\vee\mf{v})} = \bigg\{
\prod_{j=1}^\ell a_{\mf{u}(j)} \prod_{k=1}^m b_{\mf{v}(k)} 
\bigg\} \Pf \begin{bmatrix}
K_N(x_{\mf{u}(k)}, x_{\mf{u}(k')}) & K_N(x_{\mf{u}(k)},
z_{\mf{v}(n')} ) \\ K_N(z_{\mf{v}(n)} , x_{\mf{u}(k')}) &
K_N(z_{\mf{v}(n)} , z_{\mf{v}(n')}) 
\end{bmatrix},
\]
where $k$ and $k'$ are indices that run from 1 to $\ell$ and $n$ and
$n'$ are indices which run from 1 to $m$.  Thus, $\Pf(\mathbf{J} +
\mathbf{K})$ equals 
\begin{equation}
\label{eq:17}
\sum_{(\ell,m) \atop \ell + m \leq T} \sum_{\mf{u} \in
  \mf{I}_{\ell}^U} \sum_{\mf{v} \in \mf{I}_{m}^V} \bigg\{
\prod_{j=1}^\ell a_{\mf{u}(j)} \prod_{k=1}^m b_{\mf{v}(k)} 
\bigg\} \Pf \begin{bmatrix}
K_N(x_{\mf{u}(k)}, x_{\mf{u}(k')}) & K_N(x_{\mf{u}(k)},
z_{\mf{v}(n')} ) \\ K_N(z_{\mf{v}(n)} , x_{\mf{u}(k')}) & K_N(z_{\mf{v}(n)} , z_{\mf{v}(n')})
\end{bmatrix},
\end{equation}
and Theorem~\ref{thm:2} follows by equating the coefficients of 
$\prod_{j=1}^\ell a_{j} \prod_{k=1}^m b_{k}$ in (\ref{eq:26}) and (\ref{eq:17}).

\section{Kernel Asymptotics for Ginibre's Ensemble}

We now turn to the large $N$ asymptotics of the matrix kernel for
Ginibre's ensemble of real matrices.  In fact, we maintain our
restriction to the case where $N = 2M$ is even, and consider the
asymptotics of $K_{2M}$ as $M \rightarrow \infty$.  Throughout this
section we will take $\phi$ to be the function given by
\[
\phi(\gamma) = \exp(-\gamma^2/4 -\overline{\gamma}^2/4) \sqrt{ \erfc
  \big(\sqrt{2} | \ip \gamma | \big) }.
\]
Notice that, since $\erfc(0) = 1$, when $\gamma \in \R$ this reduces
to $\exp(-\gamma^2/2)$.  

The skew-orthogonal polynomials for this weight are reported in
\cite{forrester-2007} to be 
\[
\pi_{2m}(\gamma) = \gamma^{2m} \qquad \pi_{2m+1}(\gamma) = \gamma^{2m+1} - 2m \gamma^{2m-1}
\]
with normalization $\la \pi_{2m} | \pi_{2m+1} \ra_{\nu} = 2 \sqrt{2
  \pi} (2m)!$.  The skew-orthogonality of these polynomials imply that  
\begin{align}
S_{2M}(\gamma,\gamma') &=  \frac{1}{\sqrt{2 \pi}} \sum_{m=0}^{M-1} \frac{\wt
  \pi_{2m}(\gamma) \epsilon_{\lambda} \wt \pi_{2m+1}(\gamma') -  \wt \pi_{2m+1}(\gamma)
  \epsilon_{\lambda}  \wt \pi_{2m}(\gamma')}{(2m)!}, \label{eq:37} \\
DS_{2M}(\gamma,\gamma') &= \frac{1}{\sqrt{2 \pi}} \sum_{m=0}^{M-1} \frac{\wt
  \pi_{2m}(\gamma) \wt \pi_{2m+1}(\gamma') -  \wt \pi_{2m+1}(\gamma) \wt
  \pi_{2m}(\gamma')}{(2m)!}, \label{eq:38} \\
IS_{2M}(\gamma,\gamma') &= \frac{1}{\sqrt{2 \pi}} \sum_{m=0}^{M-1}
\frac{\epsilon_{\lambda} \wt \pi_{2m}(\gamma) \epsilon_{\lambda} \wt \pi_{2m+1}(\gamma') -  \epsilon_{\lambda}
  \wt \pi_{2m+1}(\gamma) \epsilon_{\lambda} \wt
  \pi_{2m}(\gamma')}{(2m)!}. \label{eq:39}
\end{align}

The correlation functions are all of the form $\Pf \mathbf{K}$ for an
appropriate matrix $\mathbf{K}$ whose entries are given in terms of
(\ref{eq:37}), (\ref{eq:38}) and (\ref{eq:39}).  If $\mathbf{D}$ is a
square matrix such that the product $\mathbf{D} \mathbf{K}
\mathbf{D}^{\transpose}$ makes sense, then $\Pf(\mathbf{D} \mathbf{K}
\mathbf{D}^{\transpose}) = \Pf \mathbf K \cdot \det \mathbf D$.  And
thus, if $\det \mathbf{D} = 1$ we have $\Pf \mathbf K = \Pf(\mathbf{D}
\mathbf{K} \mathbf{D}^{\transpose})$.  That is, we may alter (and
potentially simplify) the presentation of the Pfaffian representation
of the correlation functions by modifying $\mathbf K$ in this manner
by a matrix with determinant 1.  When $\mathbf{D}$ is diagonal, the process
of modifying $\mathbf{K}$ by $\mathbf{D}$ preserves the block 
structure of $\mathbf{K}$.  That is, the effect of modifying
$\mathbf{K}$ by $\mathbf{D}$ affects changes at the kernel level and
the correlation functions can be represented as the Pfaffian of a
block matrix ({\it cf}.~Theorem~\ref{thm:2}) with respect to a new
matrix kernel $\wt K_N$ dependent on $K_N$ and $\mathbf{D}$.  This
will allow us to write the correlation functions of Ginibre's real
ensemble in the simplest manner possible by `factoring' unnecessary
terms out of the kernel.

It will be convenient to define $c_M, s_M$ and $e_M$ to be the degree
$2M-2$ Taylor polynomials for $\cosh$, $\sinh$ and $\exp$
respectively.  Explicitly,
\[
c_M(\gamma) := \sum_{m=0}^{M-1} \frac{\gamma^{2m}}{(2m)!}, \qquad 
s_M(\gamma) := \sum_{m=1}^{M-1} \frac{\gamma^{2m-1}}{(2m-1)!}
\]
and
\[
e_M(\gamma) := c_M(\gamma) + s_M(\gamma) = \sum_{m=0}^{2M-2} \frac{\gamma^m}{m!}.
\]
\begin{thm}
\label{thm:3}
The $\ell,m$--correlation function of Ginibre's real ensemble of $2M 
\times 2M$ matrices is given by
\[
R_{\ell,m}(\mathbf{x}, \mathbf{z}) = \Pf \begin{bmatrix}
\wt K_{2M}(x_j, x_{j'}) & \wt K_{2M}(x_j, z_{k'} ) \\
\wt K_{2M}(z_k , x_{j'}) & \wt K_{2M}(z_k , z_{k'})
\end{bmatrix}; \qquad 
\begin{array}{ll}
j,j'=1,2,\ldots,\ell; \\
k,k'=1,2,\ldots,m,
\end{array}
\]
where
\[
\wt K_{2M}(\gamma, \gamma') = \begin{bmatrix} 
\wt{DS}_{2M}(\gamma, \gamma') & \wt{S}_{2M}(\gamma, \gamma') \\
-\wt{S}_{2M}(\gamma', \gamma) & \wt{IS}_{2M}(\gamma, \gamma') +
\mathcal{E}(\gamma, \gamma') 
\end{bmatrix}
\]
is given as follows.  Let $x, x' \in \R$ and $z, z' \in \C^{\ast}$,
then:
\begin{enumerate}
\item  The entries in the real/real kernel, $\wt K_{2M}(x, x')$,
  are given by  
\begin{itemize}
\item ${\displaystyle \wt S_{2M}(x, x') =
    \frac{e^{-\frac{1}{2}(x-x')^2}}{\sqrt{2\pi}} e^{-x
      x'} e_M(x x') 
  + \frac{e^{-x^2/2}}{\sqrt{2\pi}} r_M(x, x'),}$  where 
\begin{equation}
\label{eq:31}
r_M(z,x') = \frac{2^{M -  3/2}}{(2M-2)!} \sgn(x') 
z^{2M-1} \cdot \gamma\!\left(M - \frac{1}{2}, \frac{x'^2}{2} \right);
\end{equation}
\item ${\displaystyle \wt{DS}_{2M}(x, x') =
    \frac{e^{-\frac{1}{2}(x-x')^2}}{\sqrt{2\pi}} (x' -
    x) e^{-x x'} e_M(x x')}$;  
\item $\wt{IS}_{2M}(x, x') =$
\begin{align*}  
    \frac{e^{-x^2/2}}{2\sqrt{\pi}} \sgn(x')
    \int_0^{x'^2/2} \frac{e^{-t}}{\sqrt{t}} c_M(x \sqrt{2 t})
    \, dt - \frac{e^{-x'^2/2}}{2\sqrt{\pi}} \sgn(x)
    \int_0^{x^2/2} \frac{e^{-t}}{\sqrt{t}} c_M(x' \sqrt{2 t})
    \, dt.
\end{align*}
\end{itemize}
\item  The entries in the complex/complex kernel, $\wt K_{2M}(z, z')$,
  are given by  
\begin{itemize}
\item ${\displaystyle \wt S_{2M}(z,z') = \frac{i e^{-\frac{1}{2}(z -
        \overline z')^2}}{\sqrt{2 \pi}} (\overline z' - z) \sqrt{
      \erfc \big(\sqrt{2}  \ip z  \big)  \erfc
      \big(\sqrt{2}  \ip{z'}  \big) }
     e^{-z \overline z'} e_M(z \overline z');
}$
\item ${\displaystyle \wt{DS}_{2M}(z,z') = \frac{e^{-\frac{1}{2}(z -
        z')^2}}{\sqrt{2 \pi}} (z' - z) \sqrt{
      \erfc \big(\sqrt{2}  \ip z  \big)  \erfc
      \big(\sqrt{2}  \ip{z'}  \big) }
     e^{-z z'} e_M(z z'); }$
\item ${\displaystyle \wt{IS}_{2M}(z,z') =
    \frac{-e^{-\frac{1}{2}(\overline z - 
        \overline z')^2}}{\sqrt{2 \pi}} (\overline z' \! - \overline z) \sqrt{
      \erfc \big(\sqrt{2} \ip z  \big) \! \erfc
      \big(\sqrt{2} \ip{z'}  \big) }
     e^{-\overline z  \overline z'} \! e_M(\overline z \overline z'). }$
\end{itemize}

\item The entries in the real/complex kernel, $\wt K_{2M}(x, z)$, are
  given by
\begin{itemize}
\item ${\displaystyle \wt S_{2M}(x,z) = 
\frac{i e^{-\frac{1}{2}(x-\overline z)^2}}{\sqrt{2\pi}} (\overline z - x) \sqrt{
      \erfc \big(\sqrt{2} \ip z  \big)} e^{-x\overline z} e_M(x
    \overline z);
}$
\item ${\displaystyle \wt S_{2M}(z,x) = 
\frac{e^{-\frac{1}{2}(x-z)^2}}{\sqrt{2\pi}} \sqrt{
      \erfc \big(\sqrt{2} \ip z  \big)} e^{-x z} e_M(x z)}$
\[ \hspace{5cm} +
    \frac{e^{-\frac{1}{4}(z^2 + \overline z^2)}}{\sqrt{2\pi}} \sqrt{\erfc \big(\sqrt{2} \ip z
      \big)} r_M(z,x);
\]
\item ${\displaystyle \wt{DS}_{2M}(x,z) = 
\frac{e^{-\frac{1}{2}(x- z)^2}}{\sqrt{2\pi}} (z - x) \sqrt{
      \erfc \big(\sqrt{2} \ip z  \big)} e^{-x z} e_M(x z);
}$
\item ${\displaystyle \wt{IS}_{2M}(x,z) = 
\frac{e^{-\frac{1}{2}(x-\overline z)^2}}{\sqrt{2\pi}} \sqrt{
      \erfc \big(\sqrt{2} \ip z  \big)} e^{-x \overline z} e_M(x
    \overline z)}$
\[ \hspace{5cm} +
    \frac{e^{-\frac{1}{4}(z^2 + \overline z^2)}}{\sqrt{2\pi}}
    \sqrt{\erfc \big(\sqrt{2} \ip z \big)} r_M(\overline z,x).
\]
\end{itemize}

\end{enumerate}
\end{thm}

Theorem~\ref{thm:3} allows us to derive the $M \rightarrow \infty$
limit of  $\wt K_{2M}$.   
\begin{cor}
\label{cor:1}
Let $x$ and $x'$ be real numbers, and suppose $z$ and $z'$ are 
complex numbers in the open upper half plane. We define
${\displaystyle \wt K = \lim_{M \rightarrow \infty} \wt K_{2M}}$.  Then,
\begin{enumerate}
\item The limiting real/real kernel, $\wt K(x, x')$, is given by 
\[ \wt K(x, x') = 
\begin{bmatrix}
\frac{ 1}{\sqrt{2\pi}} (x' - x) e^{-\frac{1}{2}(x -
     x')^2}  & 
\frac{1}{\sqrt{2 \pi}} e^{-\frac{1}{2}(x -
      x')^2} \\
-\frac{1}{\sqrt{2 \pi}} e^{-\frac{1}{2}(x -
      x')^2}
  & \frac{1}{2} \sgn(x - x') \erfc\left(\frac{| x - x'
      |}{\sqrt{2}} \right)
\end{bmatrix}.
\]
\item The limiting complex/complex kernel, $\wt K(z,z')$, is given by
\begin{align*}
\wt K(z,z') &= \frac{1}{\sqrt{2 \pi}} \sqrt{ \erfc \big(\sqrt{2} 
\ip z  \big)  \erfc \big(\sqrt{2} 
\ip z'  \big) } \\
& \hspace{3cm} \times \begin{bmatrix}
(z'-z)  e^{-\frac{1}{2}(z-z')^2} & i(\overline z' - z)
 e^{-\frac{1}{2}(z-\overline z')^2} \\ 
i(z' - \overline{z})
  e^{-\frac{1}{2}(\overline z-z')^2} &
 -(\overline z' - \overline z)  e^{-\frac{1}{2}(\overline z-\overline z')^2} 
\end{bmatrix}.
\end{align*}
\item The limiting real/complex kernel, $\wt K(x,z)$, is given by
\[
\wt K(x,z) = \frac{1}{\sqrt{2 \pi}} \sqrt{ \erfc \big(\sqrt{2} 
\ip z  \big) } \begin{bmatrix}
(z - x) e^{-\frac{1}{2}(x - z)^2} & i(\overline z -
x)  e^{-\frac{1}{2}(x - \overline z)^2} \\ 
e^{-\frac{1}{2}(x - z)^2} & -i 
e^{-\frac{1}{2}(x - \overline z)^2}
\end{bmatrix}.
\]
\end{enumerate}
\end{cor}

We now turn to the kernel of Ginibre's Real Ensemble in the bulk.  The
circular law for $N \times N$ matrices with iid Gaussian entries says
that, when normalized by $N^{-1/2}$ the density of eigenvalues becomes
uniform on the unit disk as $N \rightarrow \infty$ (See
\cite{MR773436} for a proof of this fact when the entries are iid
Gaussian, and \cite{MR1428519} and \cite{2007arXiv0708.2895T} for more
general results).  This gives us the appropriate scaling when
considering the matrix kernel in the bulk.  Specifically, in this
section we will be interested in the large $M$ limit of $\wt K_{2M}(u
\sqrt{2M} + s, u \sqrt{2M} + s')$ where $u$ is a point in the open unit
disk, and $s$ and $s'$ are complex numbers.

When $u$ is real we expect that the limiting kernel under this scaling
should yield $\wt K(s, s')$; indeed this is the case.  When $u$ is
nonreal a new kernel arises.
\begin{thm}
\label{thm:5}
Let $-1 < u < 1$ be a real number, let $r_1, r_2, \ldots, r_{\ell} \in
\R$ and $s_1, s_2, \ldots, s_m$ be in the open upper half plane.  Set,
\begin{align*}
x_j = u \sqrt{2M} + r_j \qquad j=1,2,\ldots, \ell; \qquad \mbox{and}
\qquad 
z_k = u \sqrt{2M} + s_k \qquad k=1,2,\ldots, m. 
\end{align*}
Then,
\[
\lim_{M \rightarrow \infty} R_{\ell,m}(\mathbf{x}, \mathbf{z}) = \Pf \begin{bmatrix}
\wt K(r_j, r_{j'}) & \wt K(r_j, s_{k'} ) \\
\wt K(s_k , r_{j'}) & \wt K(s_k , s_{k'})
\end{bmatrix}; \qquad 
\begin{array}{ll}
j,j'=1,2,\ldots,\ell; \\
k,k'=1,2,\ldots,m,
\end{array} 
\]
where $\wt K$ is given as in Corollary~\ref{cor:1}.
\end{thm}

\begin{thm}
\label{thm:6}
Let $u$ be in the open upper half plane such that $|u| < 1$ and suppose $s_1, s_2, \ldots,
s_m \in \C$.  Set, 
\[
z_k = u \sqrt{2M} + s_k \qquad k=1,2,\ldots, m. 
\]
Then,
\[
\lim_{M \rightarrow \infty} R_{0,m}(-,\mathbf{z}) = \det\left[
\frac{1}{\pi}\exp\left( -\frac{|s_k|^2}{2} - \frac{|s_{k'}|^2}{2} +
  s_{k} \overline{s}_{k'} \right) \right]_{k,k'=1}^m.
\]
\end{thm}
\begin{rem}
The function
\[
\kappa(s,s') = \frac{1}{\pi}\exp\left( -\frac{|s|^2}{2} - \frac{|s'|^2}{2} +
  s \overline{s'} \right) 
\]
is the limiting kernel of Ginibre's complex ensemble.  Thus, the
limiting correlation functions in the bulk of Ginibre's real ensemble
off the real line is identical to the limiting correlation functions
in the bulk of Ginibre's complex ensemble.  
\end{rem}

\section{Proofs}

\subsection{The Proof of Lemma~\ref{lemma:4} and Lemma \ref{lemma:1}}
\label{sec:proof-lemma-refl}
In the case of the real asymmetric ensembles $Y = \C$ and $b = 1$.  In
the case of the Hermitian ensembles $Y = \R$ and $b = \sqrt{\beta}$.
We start with 
\[
\la q_n | q_{n'} \ra_{w(\lambda + \eta)} = \la \widetilde{q}_n |
\widetilde{q}_{n'} \ra_{\lambda + \eta} = \int_{Y}
\big( \widetilde{q}_n(\gamma) \, \epsilon_{\lambda +
  \eta}\widetilde{q}_{n'}(\gamma) - \epsilon_{\lambda +
  \eta}\widetilde{q}_n(\gamma) \, \widetilde{q}_{n'}(\gamma) \big) \, d(\lambda +
\eta)(\gamma). 
\]
An easy calculation reveals that this is equal to
\begin{align*}
  \la {q_n} | q_{n'} \ra_{w \lambda} &+ \frac{2}{b} \sum_{t=1}^T c_t
  \big( \widetilde{q}_n(y_t) \epsilon_{\lambda}
  \widetilde{q}_{n'}(y_t) - \widetilde{q}_{n'}(y_t) \epsilon_{\lambda}
  \widetilde{q}_n(y_u) \big) \\ & \hspace{4cm}- \frac{2}{b}
  \sum_{t=1}^T \sum_{t'=1}^T c_t c_{t'} \widetilde{q}_n(y_t)
  \widetilde{q}_{n'}(y_{t'}) \, \mathcal{E}(y_t, y_{t'}).
\end{align*}
Next we define $\mathbf{A}$ to be the $Nb \times 2T$ matrix given by
\[
\mathbf{A} :=  \left[\sqrt{\frac{2 c_t}{b}} \, \widetilde{q}_n(y_t) \quad
\sqrt{\frac{2 c_t}{b}} \,\epsilon_{\lambda} 
\widetilde{q}_n(y_t) \right]; \qquad  n = 0, 1, \ldots, Nb-1
\quad \mbox{and} \quad t = 1, 2, 
\ldots T,
\]
and the $2T \times 2T$ matrix $\mathbf{B}$,
\[
\mathbf{B} := -\mathbf{J} +  \begin{bmatrix} 
\sqrt{c_t c_{t'}} \, \mathcal{E}(y_t, y_{t'}) & 0 \\ 0 & 0
\end{bmatrix}; \qquad t,t'=1,2, \ldots, T.
\]
We define $\mathbf{C}$ to be the $Nb \times N$b matrix given by
$\mathbf{C} = (\mathbf{U}_{\mathbf{q}}^{w\lambda})^{\transpose}$; the
$n,n'$ entry of $\mathbf{C}$ is $\mu_{n,n'}$.  A bit of matrix algebra
reveals that 
\begin{equation}
\label{eq:12}
\frac{Z^{w(\lambda + \eta)}}{Z^w} = \frac{\Pf(
  \mathbf{C}^{-\mathsf{T}} - \mathbf{A B A}^{\mathsf{T}} )}{\Pf(
  \mathbf{C}^{-\mathsf{T}})} =  \frac{\Pf(
  \mathbf{B}^{-\mathsf{T}} - \mathbf{A}^{\mathsf{T}} \mathbf{B A} )}{\Pf(
  \mathbf{B}^{-\mathsf{T}})},
\end{equation}
where the second equality comes from the Pfaffian Cauchy-Binet formula
(see \ref{app:B}).  Notice that
\[
\mathbf{B}^{-\transpose} = - \mathbf{J} -    \begin{bmatrix} 
0 & 0 \\
0 &  \sqrt{c_t c_{t'}} \, \mathcal{E}(y_t, y_{t'}) 
\end{bmatrix}; \qquad t,t'=1,2, \ldots, T.
\]
And from Lemma~\ref{lemma:6}, $\Pf( \mathbf{B})^{\transpose} =
\Pf(-\mathbf{J}) = (-1)^T$.  A bit more matrix algebra reveals,  
\[
\mathbf{A}^{\transpose} \mathbf{C A} =  \begin{bmatrix} 
\sqrt{c_t c_{t'}} \, DS_N(y_t, y_{t'}) & \sqrt{c_t c_{t'}} \,
S_N(y_t, y_{t'}) \\ 
-\sqrt{c_t c_{t'}} \, S_N(y_{t'}, y_{t}) & \sqrt{c_t c_{t'}} \,
IS_N(y_t, y_{t'}) 
\end{bmatrix}; \qquad t,t'=1,2,\ldots T.
\]
Using these facts and simplifying (\ref{eq:12}) we find
\[
\frac{Z^{w(\lambda + \eta)}}{Z^w} = (-1)^T \Pf\big( -\mathbf{J} - [
  \sqrt{c_t c_{t'}} K_N(y_t, y_{t'})]\big); \qquad t,t'=1,2,\ldots,T,
\]
and the Lemma follows by using the fact that if $\mathbf{E}$ is an
antisymmetric $2T \times 2T$ matrix, then $\Pf(-\mathbf{E}) = (-1)^T
\Pf(\mathbf{E}).$

\subsection{The Proof of Lemma~\ref{lemma:5} and Lemma~\ref{lemma:2}} 
\label{sec:proof-lemma-refl-1}
We start with
\begin{equation}
\label{eq:28}
\frac{1}{L! M! 2^M} \int_{\R^L} \int_{\C^M} \Omega_{L,M}(\bs{\alpha}, \bs{\beta})
\, d(\lambda+\eta)_L(\bs{\alpha}) \,
d(\lambda+\eta)_{2M}(\bs{\beta}).
\end{equation}
Notice that in the case of the Hermitian ensembles that this is equal
to $Z^{w(\lambda + \eta)}$ when $L = N$ and $M = 0$, and the proof of
Lemma~\ref{lemma:2} follows from the proof recorded here by setting every
instance of $M$ to 0.

First we write
\begin{equation}
\label{eq:27}
d(\lambda + \eta)_L(\bs{\alpha}) = \prod_{j=1}^L
\big(d\lambda_1(\alpha_j) + d\eta_1(\alpha_j) \big) = \sum_{\ell = 0}^L
\sum_{\mf{u} \in \mf{I}_{\ell}^L}
d\eta_{\ell}(\bs{\alpha}_{\mf{u}})
d\lambda_{L-\ell}(\bs{\alpha}_{\mf{u}'}),
\end{equation}
where given $\mf{t} \in \mf{I}_n^W$, we define $\mf{t}'$ to be the
unique element in $\mf{I}_{W-n}^W$ whose range is disjoint from
$\mf{t}$.  Notice that since $\mf{u}'$ appears in the summand on the
right hand side of (\ref{eq:27}), the inner sum is not actually empty
when $\ell = 0$; in this situation the summand is equal to
$d\lambda_L(\bs{\alpha})$.  

Similarly,
\[
d(\lambda + \eta)_{2M}(\bs{\beta}) = \sum_{m = 0}^M
\sum_{\mf{v} \in \mf{I}_{m}^M}
d\eta_{2m}(\bs{\beta}_{\mf{v}})
d\lambda_{2(M-m)}(\bs{\beta}_{\mf{v}'}). 
\]
Thus, (\ref{eq:28}) equals
\begin{align*}
&\sum_{\ell = 0}^L \sum_{m = 0}^M \sum_{\mf{u} \in \mf{I}_{\ell}^L}  
\sum_{\mf{v} \in \mf{I}_{m}^M} \frac{1}{L! M! 2^M} 
\int\limits_{\R^L} \int\limits_{\C^M} \Omega_{L,M}(\bs{\alpha}, \bs{\beta})
d\eta_{\ell}(\bs{\alpha}_{\mf{u}})
d\lambda_{L-\ell}(\bs{\alpha}_{\mf{u}'}) \, d\eta_{2m}(\bs{\beta}_{\mf{v}})
d\lambda_{2(M-m)}(\bs{\beta}_{\mf{v}'}). 
\end{align*}
We can relabel the $\alpha$ and $\beta$ in the integrand in any manner
we wish, and in particular we may make the integrand independent of
$\mf{u}$ and $\mf{v}$.  In particular, if we set $\mf{i} \in
\mf{I}_n^W$ to be the identity function on $\{1,2,\ldots,n\}$, and
since the cardinality of $\mf{I}_n^W$ is ${W \choose n}$, we find that
(\ref{eq:28}) is equal to 
\begin{align}
&   \sum_{\ell = 0}^L \sum_{m = 0}^M \frac{1}{\ell! (L-\ell)! m! (M -
  m)! 2^M} \nonumber \\
& \quad \times \int\limits_{\R^{L-\ell}} \int\limits_{\C^{M-m}} \bigg\{
\int_{\R^{\ell}} \int_{\C^m} \Omega_{L,M}(\bs{\alpha}, \bs{\beta})
 \, d\eta_{\ell}(\bs{\alpha}_{\mf{i}})
 \, d\eta_{2m}(\bs{\beta}_{\mf{i}}) \bigg\}
 d\lambda_{L-\ell}(\bs{\alpha}_{\mf{i}'}) \,
 d\lambda_{2(M-m)}(\bs{\beta}_{\mf{i}'}).   \label{eq:30}
\end{align}
Now,
\[
d\eta_{\ell}(\bs{\alpha}_{\mf{i}}) = \prod_{j=1}^{\ell}
  \eta_1(\alpha_j) = \prod_{j=1}^{\ell} \sum_{u=1}^U a_u
  \, d\delta(\alpha_j - x_u). 
\]
We may exchange the sum and the integral on the right hand side of
this expressions by using the set, $\mf{F}_{\ell}^U$ of all functions
from $\{1, 2, \ldots, \ell\}$ into $\{1,2,\ldots,U\}$. Specifically,
\[
d\eta_{\ell}(\bs{\alpha}_{\mf{i}}) = \sum_{\mf{u} \in \mf{F}_{\ell}^U}
\bigg\{\prod_{j=1}^{\ell} a_{\mf{u}(j)}  \, d\delta(\alpha_j -
x_{\mf{u}(j)})\bigg\}, 
\]
and similarly,
\[
d\eta_{2m}(\bs{\beta}_{\mf{i}}) = \sum_{\mf{v} \in \mf{F}_{m}^V}
\bigg\{ \prod_{k=1}^{m} b_{\mf{v}(k)}\, d\widehat{\delta}(\beta_k -
z_{\mf{v}(k)})\bigg\}. 
\]
Thus,
\begin{align}
  &\int_{\R^{\ell}} \int_{\C^m} \Omega_{L,M}(\bs{\alpha}, \bs{\beta})
   \,   d\eta_{\ell}(\bs{\alpha}_{\mf{i}}) \,
  d\eta_{2m}(\bs{\beta}_{\mf{i}})  \nonumber \\ & \qquad = \sum_{\mf{u} \in
    \mf{F}_{\ell}^U} \sum_{\mf{v} \in \mf{F}_{m}^V} \int_{\R^{\ell}}
  \int_{\C^m} \Omega_{L,M}(\bs{\alpha}, \bs{\beta})
  \,\bigg\{\prod_{j=1}^{\ell} a_{\mf{u}(j)} \, 
  d\delta(\alpha_j - x_{\mf{u}(j)})\bigg\} \bigg\{ \prod_{k=1}^{m}
  b_{\mf{v}(k)}\, d\widehat{\delta}(\beta_k - z_{\mf{v}(k)}) \bigg\} \nonumber \\
  & \qquad = \sum_{\mf{u} \in \mf{F}_{\ell}^U} \sum_{\mf{v} \in
    \mf{F}_{m}^V} \bigg\{ \prod_{j=1}^{\ell} a_{\mf{u}(j)}
  \prod_{k=1}^{m} b_{\mf{v}(k)} \bigg\} 2^m \Omega_{L,M}(\mathbf{x}_{\mf{u}}
  \vee  \bs{\alpha}_{\mf{i}'}, \mathbf{z}_{\mf{v}} \vee
  \bs{\beta}_{\mf{i}'}). \label{eq:29}
\end{align}
Notice that if $\mf{u}$ or $\mf{v}$ is not one-to-one then
$|\Delta(\mathbf{x}_{\mf{u}} \vee  \bs{\alpha}_{\mf{i}},
\mathbf{z}_{\mf{v}} \vee  \bs{\beta}_{\mf{i}})| = 0$.  We may consequently
replace the sums over $\mf{F}_{\ell}^U$ and $\mf{F}_k^V$ with their
respective subsets of one-to-one functions.  Moreover, since
$\Omega_{L,M}$ is symmetric in the coordinates of each of its
arguments, we may replace each one-to-one function in these sums with
the increasing function with the same range so long as we compensate
by multiplying by $\ell!$ and $m!$.  Lemma~\ref{lemma:5} follows from
the definition of $R_{\ell,m,L,M}$ by substituting (\ref{eq:29}) into
(\ref{eq:30}).  Lemma~\ref{lemma:2} follows from the fact that $R_n =
R_{n,0,N,0}/Z^w$.

\subsection{The Proof of Theorem~\ref{thm:3} and Corollary~\ref{cor:1}}  

It shall be convenient to introduce the following variants of $S_{2M}$
and $DS_{2M}$: 
\[
\widehat S_{2M}(\gamma, \gamma') := \sum_{m=0}^{M-1} \frac{\pi_{2m}(\gamma)
  \epsilon_{\lambda} \wt \pi_{2m+1}(\gamma') - \pi_{2m+1}(\gamma)
  \epsilon_{\lambda} \wt \pi_{2m}(\gamma')}{(2m)!},
\]
and
\[
\widehat{DS}_{2M}(\gamma,\gamma') := \sum_{m=0}^{M-1} \frac{
  \pi_{2m}(\gamma) \pi_{2m+1}(\gamma') -  \pi_{2m+1}(\gamma)
  \pi_{2m}(\gamma')}{(2m)!},  
\]
The following lemma gives a closed form for these functions.

\begin{lemma}
\label{lemma:3} Let $x$ be a real number, and suppose $z$ and $z'$
are complex numbers.
\begin{enumerate}
\item $\widehat S_{2M}(z, x) =  \phi(x) e_M(z x) +
  r_M(z,x),$ 
  where $r_M$ is given in (\ref{eq:31});
\item $\widehat{DS}_{2M}(z,z') = (z'-z) e_M(zz').$
\end{enumerate}
\end{lemma}
\begin{proof}
First we compute $\epsilon_{\lambda} \wt \pi_{2m}$ and
$\epsilon_{\lambda} \wt \pi_{2m+1}$.  We start by noticing,
\[
\epsilon_{\lambda} \wt g(x) = \frac{1}{2} \int_{\infty}^{\infty}
g(y) \sgn(y - x) \, dy = -\frac{1}{2} \int_{-\infty}^x g(y) \,
dy + \frac{1}{2} \int_{x}^{\infty} g(y) \, dy.
\]
When $g(y) = e^{-y^2/2} y^n$, we may evaluate the latter two
integrals in terms of the incomplete gamma functions.  
\[
\epsilon_{\lambda} \wt g(x) = \left\{
\begin{array}{ll}
-2^{(n-1)/2} \sgn(x) \cdot \gamma\!\left( \frac{n+1}{2},
  \frac{x^2}{2} \right) & \quad \mbox{if } n \mbox{ is even;} \\ & \\
2^{(n-1)/2} \Gamma\!\left( \frac{n+1}{2},
  \frac{x^2}{2} \right) & \quad \mbox{if } n \mbox{ is odd.} \\
\end{array}
\right.
\]
We immediately conclude that 
\begin{equation}
\label{eq:18}
\epsilon_{\lambda} \wt \pi_{2m}(x) = -2^{m-1/2} \sgn(x) \cdot
\gamma\!\left( m+ \frac{1}{2}, \frac{x^2}{2} \right),
\end{equation}
and
\begin{equation}
\label{eq:20}
\epsilon_{\lambda} \wt \pi_{2m+1}(x) = 2^m \left[\Gamma\!\left(m+1,
    \frac{x^2}{2}\right) - m \Gamma\!\left(m,
    \frac{x^2}{2}\right) \right] = x^{2m} e^{-x^2/2},
\end{equation}
where in the second equality we used the fact that $\Gamma(a+1, x) = a
\Gamma(a,x) + x^a e^{-x}$. 

Using (\ref{eq:18}) and (\ref{eq:20}), we may write
\begin{align*}
&\widehat S_{2M}(z,x) = \phi(x) c_M(z x) \\
& \quad + \sgn(x) \left\{
  \sum_{m=0}^{M-1} \frac{2^{m-1/2}}{(2m)!} z^{2m+1} \gamma\!\left( m+
    \frac{1}{2}, \frac{x^2}{2} \right) - \sum_{m=1}^{M-1}
  \frac{2^{m-1/2}}{(2m-1)!} z^{2m-1} \gamma\!\left( m+ \frac{1}{2},
    \frac{x^2}{2} \right) \right\}.
\end{align*}
Next, we use the fact that
\begin{equation}
\label{eq:21}
\gamma(a+1,x) = a \gamma(a,x) - x^a e^{-x},
\end{equation}
so that the second sum in this expression becomes 
\[
- \sum_{m=1}^{M-1} \frac{2^{(m-1)-1/2}}{(2(m-1))!} z^{2(m-1)+1}
\gamma\!\left( (m-1)+ \frac{1}{2}, \frac{x^2}{2} \right) +
\phi(x) \sum_{m=1}^{M-1} \frac{z^{2m-1} |x|^{2m-1}}{(2m-1)!}.
\]
Consequently,
\begin{align*}
\widehat S_{2M}(z,x) &= \phi(x) \big( c_M(z x) + \sgn(x)
s_M(z |x|) \big) + \frac{2^{M - 3/2}}{(2M-2)!} \sgn(x)
z^{2M-1} \, \gamma\!\left(M - \frac{1}{2}, \frac{x^2}{2} \right) \\
&= \phi(x) e_M(z x) + \frac{2^{M - 3/2}}{(2M-2)!} \sgn(x)
z^{2M-1} \, \gamma\!\left(M - \frac{1}{2}, \frac{x^2}{2} \right).
\end{align*}

Turning to $\widehat{DS}_{2M}$,
\begin{align*}
  \widehat{DS}_{2M}(z,z') &= \sum_{m=0}^{M-1} \frac{z^{2m}(z'^{2m+1} - 2m
    z'^{2m-1}) - (z^{2m+1} - 2m z^{2m-1})  z'^{2m}}{(2m)!} \\
  &= \sum_{m=0}^{M-1} \frac{z^{2m} z'^{2m+1} - z^{2m+1} z'^{2m}}{(2m)!}
  + \sum_{m=1}^{M-1} \frac{z^{2m-1} z'^{2m} - z^{2m}
    z'^{2m-1}}{(2m-1)!}  \\
  &= (z' - z) \left\{ \sum_{m=0}^{M-1} \frac{z^{2m} z'^{2m}}{(2m)!} +
    \sum_{m=1}^{M-1}
    \frac{z^{2m-1} z'^{2m-1}}{(2m-1)!} \right\} \\
  & = (z'-z) e_M(z'z). \qedhere
\end{align*}
\end{proof}

With a closed form for $\widehat S_{2M}$ and $\widehat{DS}_{2M}$ in hand, we are
ready to prove Theorem~\ref{thm:3}.
\begin{proof}[Proof of Theorem~\ref{thm:3}]
From Lemma~\ref{lemma:3} we have that 
\[
S_{2M}(x, x') = \frac{\phi(x)}{\sqrt{2\pi}} \widehat S_{2M}(x, x') =
\frac{\phi(x) \phi(x')}{\sqrt{2\pi}} e_M(x x') +
\frac{\phi(x)}{\sqrt{2\pi}} r_M(x,x') 
\]
and
\[
DS_{2m}(x, x') = \frac{\phi(x) \phi(x')}{\sqrt{2\pi}}
\widehat{DS}_{2M}(x,x') = \frac{\phi(x) \phi(x')}{\sqrt{2\pi}} (x'-x) e_M(xx').
\]
Now,
\[
\phi(x) \phi(x') = e^{-\frac{1}{2}(x^2 + x'^2)} =
e^{-\frac{1}{2}(x-x')^2} e^{-xx'},
\]
and therefore,
\[
S_{2M}(x, x') = \frac{e^{-\frac{1}{2}(x-x')^2}}{\sqrt{2\pi}} e^{-xx'} 
e_M(x x') + \frac{\phi(x)}{\sqrt{2\pi}} r_M(x,x'),
\]
and
\[
DS_{2M}(x, x') = \frac{e^{-\frac{1}{2}(x-x')^2}}{\sqrt{2\pi}} (x'-x) e^{-xx'}
e_M(x x') 
\]

The computation of $IS_{2M}(x, x')$ is a bit more involved.
From~(\ref{eq:39}), (\ref{eq:18}) and (\ref{eq:20}), we see
\begin{align*}
IS_{2M}(x, x') &= 
\frac{1}{2 \sqrt{\pi}} \sum_{m=0}^{M-1}
\frac{2^m}{(2m)!} \cdot  \gamma\!\left(m + \frac{1}{2},
  \frac{x'^2}{2} \right) \sgn(x') x^{2m} e^{-x^2/2}
\\ & \hspace{3cm} -
\frac{1}{2 \sqrt{\pi}} \sum_{m=0}^{M-1}
\frac{2^m}{(2m)!} \cdot  \gamma\!\left(m + \frac{1}{2},
  \frac{x^2}{2} \right) \sgn(x) x'^{2m} e^{-x'^2/2}.
\end{align*}
$IS_{2M}(x, x')$ is clearly skew-symmetric in its arguments;
looking at the first sum in this expression we thus find,
\begin{align*}
  & \frac{1}{2 \sqrt{\pi}} \sum_{m=0}^{M-1} \frac{2^m}{(2m)!}
  \sgn(x') x^{2m} e^{-x^2/2}
  \int_0^{x'^2/2} t^{m-1/2} e^{-t} \, dt  \\
  & \qquad = \frac{1}{2 \sqrt{\pi}} e^{-x^2/2} \sgn(x')
  \int_0^{x'^2/2} \frac{e^{-t}}{\sqrt{t}} \left\{\sum_{m=0}^{M-1}
    \frac{2^m}{(2m)!}  t^{m} x^{2m}\right\} \, dt,
\end{align*}
where on the left hand side we have replaced the lower incomplete
gamma function with its integral definition, and on the right hand
side we have exploited the linearity of the integral. The sum on the
right hand side of this equation is equal to $c_M(x \sqrt{2 t})$, and
thus
\begin{align*}
& IS_{2M}(x,x') = \\
& \quad    \frac{e^{-x^2/2}}{2\sqrt{\pi}} \sgn(x')
    \int_0^{x'^2/2} \frac{e^{-t}}{\sqrt{t}} c_M(x \sqrt{2 t})
    \, dt - \frac{e^{-x'^2/2}}{2\sqrt{\pi}} \sgn(x)
    \int_0^{x^2/2} \frac{e^{-t}}{\sqrt{t}} c_M(x' \sqrt{2 t})
    \, dt.
\end{align*}

Turning to the complex/complex entries of $K_{2M}$, if $z$ is assumed
to be in the open upper half plane then $\epsilon_{\lambda} \pi_{n}(z) = i
\pi_{n}(\overline{z})$.  From this we see  
that $S_{2M}(z,z') = i \phi(z) \phi(z') \widehat{DS}_{2M}(z, z')$,
$DS_{2M}(z,z') = \phi(z) \phi(z') \widehat{DS}_{2M}(z, z')$ and $IS_{2M}(z,z')
= -\phi(z) \phi(z') \widehat{DS}_{2M}(\overline{z}, \overline{z}')$. 

Next, we define
\[
\psi(z) = e^{\frac{1}{4}(z^2 - \overline z^2)}.
\]
Notice that
\[
e^{-\frac{1}{4}(z^2 + \overline{z}^2)} e^{-\frac{1}{4}(z^2 +
  \overline z'^2)} = \psi(z) \psi(z') e^{-\frac{1}{2}(z-z')^2} e^{-zz'},
\]
and thus 
\[
\phi(z) \phi(z') = \psi(z) \psi(z') e^{-\frac{1}{2}(z-z')^2} \sqrt{ \erfc
  \big(\sqrt{2} \; \ip z  \big)  \erfc \big(\sqrt{2} \;
\ip z'  \big) } e^{-zz'}.
\]
Using this and Lemma~\ref{lemma:3}, we conclude that 
\[
S_{2M}(z,z') = \psi(z) \psi(\overline z') \frac{i
  e^{-\frac{1}{2}(z-\overline z')^2}}{\sqrt{2\pi}} (\overline z'-z) \sqrt{ \erfc 
  \big(\sqrt{2} \ip z  \big)  \erfc \big(\sqrt{2} \ip z'  \big) }
e^{-z\overline z'} e_M(z\overline z'),
\]
\[
DS_{2M}(z,z') = \psi(z) \psi(z')
\frac{e^{-\frac{1}{2}(z-z')^2}}{\sqrt{2\pi}} (z'-z) \sqrt{ \erfc
  \big(\sqrt{2} \ip z  \big)  \erfc \big(\sqrt{2} 
\ip z'  \big) } e^{-zz'} e_M(zz'),
\]
and
\[
S_{2M}(z,z') = \psi(\overline z) \psi(\overline z') \frac{-
  e^{-\frac{1}{2}(\overline z-\overline z')^2}}{\sqrt{2\pi}}
(\overline z'-z \overline ) \sqrt{ \erfc  
  \big(\sqrt{2} \ip z  \big)  \erfc \big(\sqrt{2} \ip z'  \big) }
e^{-\overline z\overline z'} e_M(\overline z\overline z').
\]

Lastly we look at the real/complex entries of $K_{2M}$.  As in all
other cases, 
\[
DS_{2M}(x, z) = \frac{\phi(x) \phi(z)}{\sqrt{2\pi}}
\widehat{DS}_{2M}(x, z),
\]
and it is easily verified that
\[
\phi(x) \phi(z) = \psi(z) e^{-\frac{1}{2}(x - z)^2} \sqrt{ \erfc
  \big(\sqrt{2} \; \ip z  \big)} e^{-xz}.
\]
Thus,
\[
DS_{2M}(x,z) = \psi(z) \frac{e^{-\frac{1}{2}(x - z)^2}}{\sqrt{2\pi}}
(z-x) \sqrt{ \erfc \big(\sqrt{2} \; \ip z  \big)} e^{-xz} e_M(xz). 
\]
Since,
\[
S_{2M}(x, z) = i \frac{\phi(x) \phi(z)}{\sqrt{2\pi}}
\widehat{DS}_{2M}(x, 
\overline z), \qquad \qquad S_{2M}(z, x) =
\frac{\phi(z)}{\sqrt{2\pi}} \widehat S_{2M}(z, x).
\]
and
\[
IS_{2M}(x, z) = -i \frac{\phi(z)}{\sqrt{2\pi}} \widehat S_{2M}(z,x),
\]
It follows that
\[
S_{2M}(x,z) = 
\psi(\overline z) \frac{i e^{-\frac{1}{2}(x-\overline
    z)^2}}{\sqrt{2\pi}} (\overline z - x) \sqrt{ 
      \erfc \big(\sqrt{2} \ip z  \big)} e^{-x\overline z} e_M(x
    \overline z),
\]
\begin{align*}
& S_{2M}(z,x) = 
\psi(z) \left\{\frac{e^{-\frac{1}{2}(x-z)^2}}{\sqrt{2\pi}} \sqrt{
      \erfc \big(\sqrt{2} \ip z  \big)} e^{-x z} e_M(x z) \right. \\
& \hspace{6cm} \left. + \frac{e^{-\frac{1}{4}(z^2 + \overline z^2)}}{\sqrt{2\pi}} \sqrt{\erfc
    \big(\sqrt{2} \ip z \big)} r_M(z,x) \right\}
\end{align*}
and
\begin{align*}
& {IS}_{2M}(x,z) = \psi(\overline{z}) \left\{
\frac{e^{-\frac{1}{2}(x-\overline z)^2}}{\sqrt{2\pi}} 
\sqrt{\erfc \big(\sqrt{2} \ip z  \big)} e^{-x \overline z} e_M(x
    \overline z) \right. \\
& \hspace{6cm} \left. + \frac{e^{-\frac{1}{4}(z^2 + \overline z^2)}}{\sqrt{2\pi}}
  \sqrt{\erfc \big(\sqrt{2} \ip z \big)} r_M(\overline z,x) \right\}.
\end{align*}

Clearly, $\psi(x) = \psi(x') = 1$, and thus we find that 
\[
K_N(\gamma, \gamma') = \begin{bmatrix}
\psi(\gamma) & 0 \\
0 & \psi(\overline \gamma) 
\end{bmatrix} \wt K_N(\gamma, \gamma') \begin{bmatrix}
\psi(\gamma') & 0 \\
0 & \psi(\overline \gamma') 
\end{bmatrix}.
\]
It follows that, if we define $\wt{\mathbf K}$ to be the matrix 
\[
\wt{\mathbf K} = \begin{bmatrix}
\wt K_N(x_j, x_{j'}) & \wt K_N(x_j, z_{k'} ) \\
\wt K_N(z_k , x_{j'}) & \wt K_N(z_k , z_{k'})
\end{bmatrix}; \qquad 
\begin{array}{ll}
j,j'=1,2,\ldots,\ell; \\
k,k'=1,2,\ldots,m,
\end{array}
\]
and $\mathbf{D}$ to be the diagonal matrix 
\[
\mathbf{D} = \mathrm{diag}\big(\psi(x_1), \psi(\overline{x_1}), \ldots,
\psi(x_{\ell}), \psi(\overline{x_{\ell}}), \psi(z_1),
\psi(\overline{z_1}) \ldots, \psi(z_m), \psi(\overline{z_m})\big),
\]
then 
\[
R_{\ell, m}(\mathbf x, \mathbf z) = \Pf( \mathbf{D} \wt{\mathbf{K}} \mathbf{D}).
\]
But, since $\psi(\overline z) = \psi(z)^{-1}$, we have that $\det
\mathbf D = 1$, and $R_{\ell, m}(\mathbf x, \mathbf z) = \Pf
\wt{\mathbf{K}}$ as claimed. 
\end{proof}

\begin{proof}[Proof of Corollary~\ref{cor:1}]

We first make use of the fact that 
\[
\lim_{M \rightarrow \infty} e_M(z) = e^z
\]
pointwise on $\C$.  This simplifies all terms in the kernel except the 
$IS_{2M}$ term.  It remains to show that $r_M(z, x) \rightarrow
0$ as $M \rightarrow \infty,$ and that
\begin{equation}
\label{eq:32}
\frac{1}{2} \sgn(x - x') + \lim_{M \rightarrow \infty}
IS_{2M}(x, x') = \frac{1}{2} \sgn(x - x')
\erfc\left(\frac{| x - x' |}{\sqrt{2}} \right).
\end{equation}
The first of these facts is easily seen by noting that $\gamma(M -
1/2, x^2/2) < \Gamma(M-1/2)$, and hence 
\[
|r_M(z,x)| < \frac{|z|^{2M-1} \; 2^{M-3/2} \;
  \Gamma(M-1/2)}{\Gamma(2M-1)} 
=  \frac{|z|^{2M-1} \; 2^{-M+3/2} \; \sqrt{\pi} \;
  \Gamma(2M-2)}{\Gamma(2M-1) \Gamma(M-1)},
\]
where we arrive at the latter equality we have used Legendre's
duplication formula on the $\Gamma(M - 1/2)$ term.  Thus,
\[
|r_M(z,x)| < \frac{|z|^{2M-1}}{\Gamma(M-1)} \cdot \frac{2^{-M+3/2} \;
  \sqrt{\pi}}{(2M-2)},
\]
and it is easy to see that both terms in this product go to 0 as $M
\rightarrow \infty$, independent of the value of $z$. 

To establish (\ref{eq:32}) we start with 
\begin{equation}
\label{eq:35}
I_M(x, x') := \int_0^{x'^2/2} \frac{e^{-t}}{\sqrt{t}}
c_M(x \sqrt{2 t}) \, dt
\end{equation}
Since the terms in $c_M$ are all positive, from the Monotone
Convergence Theorem, 
\begin{align}
I(x, x') :=& \lim_{M \rightarrow \infty} I_M(x, x') =
\int_0^{x'^2/2} \frac{e^{-t}}{\sqrt{t}} \cosh(x \sqrt{2 t}) \,
dt \nonumber \\
=& \frac{\sqrt{\pi}}{2} e^{x'^2/2} \left[ 
\erf\left( \frac{|x'| + x}{\sqrt{2}} \right) - 
\erf\left( \frac{|x'| - x}{\sqrt{2}} \right)
\right].
\label{eq:34}
\end{align}
The latter equality follows from the fact that
\[
\int_0^x \frac{e^{-t}}{\sqrt{t}} \cosh(a \sqrt{t}) \, dt =
\frac{\sqrt \pi}{2} e^{a^2/4} \left[ 
\erf\left( \frac{a}{2} + \sqrt{x} \right) - 
\erf\left( \frac{a}{2} - \sqrt{x} \right)
\right],
\]
which can be verified via differentiation.

Now, 
\begin{align*}
\lim_{M \rightarrow \infty} IS_{2M}(x, x') &= \frac{1}{2
    \sqrt \pi} \left\{ e^{-x^2/2} \sgn(x') I(x, x') -
    e^{-x'^2/2} \sgn(x) I(x',x)
  \right\} \\
&= \frac{1}{4} \left[\sgn(x') \erf\left( \frac{|x'| +
        x}{\sqrt{2}} \right) - \sgn(x') \erf\left(
      \frac{|x'| - x}{\sqrt{2}} \right) \right. \\
& \hspace{3cm}\left. - \sgn(x)
    \erf\left( \frac{x' + |x|}{\sqrt{2}} \right) + 
\sgn(x) \erf\left( \frac{x' - |x|}{\sqrt{2}} \right)
 \right]  \\
&= -\frac{1}{2} \erf\left( \frac{x - x'}{\sqrt{2}} \right).
\end{align*}
It follows that (\ref{eq:32}) can be written as
\[
\frac{1}{2}\sgn(x - x') -\frac{1}{2} \erf\left( \frac{x -
    x'}{\sqrt{2}} \right) = \frac{1}{2} \sgn(x - x') -
\frac{1}{2} \sgn(x - x') \erf\left( \frac{|x -
    x'|}{\sqrt{2}} \right),
\]
where we have exploited the fact that $\erf$ is an odd function.  We
arrive at the form for $IS_{2M}$ stated in the corollary using the
fact that $\erfc = 1 - \erf$.  
\end{proof}

\subsection{The Proof of Theorem~\ref{thm:5} and Theorem~\ref{thm:6}} 

In order to prove Theorems~\ref{thm:5} and \ref{thm:6}, it is
necessary to investigate the asymptotics of the partial sums the
exponential function.   
\begin{lemma}
\label{lemma:7}
Let $u \neq \pm 1$ be a complex number, and let $( v_m
)_{m=1}^{\infty}$ be a sequence of complex numbers satisfying 
\[
v_M = u^2 + O\left(1/\sqrt{M}\right) \qquad \mbox{as} \qquad M
\rightarrow \infty.
\]
Then, as $M \rightarrow \infty$,
\[
e^{-2M v_M} e_M(2M v_M) \sim 1 - \frac{e^{-2(1-u^2)}}{2 \pi u^2(1-u^2)}
\cdot \frac{e^{2M(1-u^2)} u^{4M}}{\sqrt{M}}.
\]
In particular, when $-1 < u < 1$ is real,  
\[
\lim_{M \rightarrow \infty} e^{-2M v_M} e_M(2M v_M) = 1.
\]
\end{lemma}
\begin{proof}
Set $v = v_M$.  We start by writing $2Mv = v(2M-2) + 2v$.  Thus,
\[
e^{-2M v} e_M(2M v) = \exp\left(-(2M-2)\left(v +
    \frac{v}{M-1}\right) \right) e_M\left((2M-2)\left(v +
    \frac{v}{M-1}\right) \right).
\]
We write 
\[
w = w_M = v + \frac{v}{M-1}.
\]
Clearly $w_M = u^2 + O(1/\sqrt M)$.  Under this hypothesis, and since
$u^2 \neq 1$ , equations (2.9), (2.15) and (1.7) of \cite{MR1113918}
imply that 
\[
e^{-(2M-2) w} e_M((2M-2) w) \sim 1 - \frac{e^{-2(1-u^2)}}{2 \pi u^2(1-u^2)}
\cdot \frac{e^{2M(1-u^2)} u^{4M}}{\sqrt{M}}. \qedhere
\]
The second statement of the lemma follows from the fact that if $u$ is
real and $|u| < 1$, then 
\[
e^{2M(1-u^2)} u^{4M} = e^{2M(1 - u^2 + 2 \log|u|)},
\]
and $1 - u^2 + 2 \log|u|$ is negative when $u$ is in $(-1,1)$.
\end{proof}

\begin{lemma}
\label{lemma:8}
Suppose $u,s \in \R$ with $|u| < 1$ and suppose $s' \in \C$.  Let
\begin{align*}
& R_M = \exp\left(-Mu^2 -u\sqrt{2M} \rp{s'} \right)
\frac{2^{M-3/2}}{(2M-2)!} (u \sqrt{2M} + s')^{2M-1} \\
& \hspace{6cm} \times \gamma\left( M - \frac{1}{2}, M u^2 + u \sqrt{2M} s +
  \frac{s^2}{2} \right). 
\end{align*}
Then, ${\displaystyle \lim_{M \rightarrow \infty} R_M = 0}.$ 
\end{lemma}
\begin{proof}
Clearly,
\[
\gamma\left( M - \frac{1}{2}, M u^2 + u \sqrt{2M} s +  \frac{s^2}{2}
\right)  \leq \Gamma\left(M-\frac{1}{2} \right) = \sqrt{\pi} 2^{-2M+3}
\frac{\Gamma(2M-2)}{\Gamma(M-1)},
\]
where the second equation follows from Legendre's duplication formula
for $\Gamma$.  Now,
\[
|u \sqrt{2M} + s'|^{2M-1} \sim |u|^{2M-1} 2^{M - 1/2} M^{M-1/2}.  
\]
Thus,
\[
|R_M| \sim \sqrt{\pi} |u|^{2M-1} \frac{M^{M+1/2}}{M!}
\exp\left(-Mu^2\right) \exp\big(-u\sqrt{2M} s\big).
\]
Using Stirling's approximation for $M!$, we find
\begin{align*}
& |R_M| \sim \frac{\sqrt 2}{2} |u|^{2M-1} \exp\big(M(1-u^2)\big)
\exp\big(-u\sqrt{2M}s\big)  \\
& \hspace{3cm} = \frac{\sqrt 2}{2 |u|} \exp\big(M(2 \log|u| + 1 -
u^2)\big) \exp\big(-u\sqrt{2M}s\big).
\end{align*}
We observe that $2 \log|u| + 1 - u^2 < 0$ since $|u|<1$, and hence
$R(M) \rightarrow 0$ as $M \rightarrow \infty$. 
\end{proof}

\begin{proof}[Proof of Theorem~\ref{thm:6}]
  Let $u$ be a point in the open upper half plane with modulus less
  than 1, and suppose $s$ and $s'$ are complex numbers.  For all but
  finitely many values of $M$, $z = u\sqrt{2M} + s$ and $z' = u
  \sqrt{2M} + s'$ are in $\C_{\ast}$.  Thus, in this case, we need
  only consider the asymptotics of the complex/complex kernel under
  these substitutions.

We will make use of the fact that if $x$ is a real number,
\[
\erfc(x) = \frac{1}{\sqrt{\pi}} \Gamma\left( \frac{1}{2}, x^2 \right)
\sim \frac{e^{-x^2}}{\sqrt{\pi} | x|}.
\]
Consequently,
\begin{equation}
\label{eq:33}
\sqrt{\erfc\big(\sqrt{2}\ip{u \sqrt{2M} + s}\big)} \sim \frac{\exp\big(-2M
  \ip{u}^2-2\sqrt{2M} \ip u \ip s - \ip s^2\big)}{\sqrt{2 \ip u}
  \sqrt[4]{M \pi}}.
\end{equation}

Now, by Theorem~\ref{thm:3},
\begin{align*}
& \wt{DS}_{2M}(u\sqrt{2M} + s, u\sqrt{2M} + s') = \frac{(s' - s)}{\sqrt{2
    \pi}} e^{-\frac{1}{2}(s - s')^2} \\
& \qquad \times \sqrt{\erfc\big(\sqrt{2}\ip{u \sqrt{2M} + s}\big)
  \erfc\big(\sqrt{2}\ip{u \sqrt{2M} + s'}\big) }\\
& \qquad \times \exp\!\left(-2 M u^2 - u \sqrt{2M}(s + s') - ss' \right)
e_M\!\left(2 M u^2 + u \sqrt{2M}(s + s') + ss' \right). \nonumber
\end{align*}
Therefore, by Lemma~\ref{lemma:7} and (\ref{eq:33}),
\begin{align*}
& \wt{DS}_{2M}(u\sqrt{2M} + s, u\sqrt{2M} + s') \sim \frac{(s' - s)}{\sqrt{2
    \pi}} e^{-\frac{1}{2}(s - s')^2} \\
& \hspace{4cm} \times 
\frac{e^{-4M \ip u^2} e^{-2 \sqrt{2M} \ip u (\ip s + \ip{s'})}
  e^{- \ip s^2 - \ip{s'}^2} 
}{2 \ip u \sqrt{M \pi}}
\\
& \hspace{4cm} \times 
\left(1 - \frac{e^{-2(1-u^2)}}{2 \pi u^2(1-u^2)}
\cdot \frac{e^{2M(1-u^2)} u^{4M}}{\sqrt{M}}\right). 
\end{align*}
It is easily seen that 
\[
\lim_{M \rightarrow \infty}\frac{e^{-4M \ip u^2} e^{-2 \sqrt{2M} \ip u (\ip s + \ip{s'})}
}{2 \ip u \sqrt{M \pi}} = 0,
\]
and,
\[
\left| e^{-4M \ip u^2} e^{2M(1-u^2)} u^{4M} \right| = e^{2M(1 - |u|^2
  + 2 \log|u|)}.
\]
Since $|u| < 1$, we have $1 - |u|^2 + 2 \log|u| < 0$, and therefore
\[
\lim_{M \rightarrow \infty} e^{-2 \sqrt{2M} \ip u (\ip s + \ip{s'})}
\left| e^{-4M \ip u^2} e^{2M(1-u^2)} u^{4M} \right|  = 0.
\]
We conclude that 
\[
\lim_{M \rightarrow \infty} \wt{DS}_{2M}(u\sqrt{2M} + s, u\sqrt{2M} +
s') = 0.
\]
And, since $\wt{IS}_{2M}(z,z') = -\wt{DS}_{2M}(\overline z, \overline
z')$, 
\[
\lim_{M \rightarrow \infty} \wt{IS}_{2M}(u\sqrt{2M} + s, u\sqrt{2M} +
s') = 0.
\]

Turning to $\wt{S}_{2M}(u\sqrt{2M} + s, u\sqrt{2M} + s')$, we set 
\[
\eta_M(s) = \exp\left(-2i \sqrt{2M} \ip u \rp s \right).
\]
From Theorem~\ref{thm:3}, we have 
\begin{align*}
&\eta_M(s) \eta_M(-s') \wt{S}_{2M}(u\sqrt{2M} + s, u\sqrt{2M} + s') =
\frac{-i}{\sqrt{2 \pi}} 
\left[ 2i \sqrt{2M} \ip{u}  + (s - \overline s') \right] \\
& \qquad \times \exp\big( 2 i \sqrt{2M} \ip u \rp s 
\big) \exp \big(- 2 i \sqrt{2M} \ip u \rp{s'} \big) \\
& \qquad \times \exp\left(4M (\ip{u})^2 - 2i\sqrt{2M} \ip{u} (s -
\overline{s}') - \frac{1}{2}(s - \overline s')^2  \right) \nonumber \\
& \qquad \times \sqrt{ \erfc\big| 2 \sqrt{M} \ip u + \sqrt{2} \ip s
  \big| \erfc\big| 2 \sqrt{M} \ip u + \sqrt{2} \ip{s'} \big|
} \nonumber \\
& \qquad \times \exp\left(-2M|u|^2 - (s\overline u + \overline s'
u)\sqrt{2M} - s \overline s'\right) e_M\left(2M|u|^2 + (s\overline u + \overline s'
u)\sqrt{2M} + s \overline s'\right). \nonumber
\end{align*}
From Lemma~\ref{lemma:7} and (\ref{eq:33}) we find that 
\begin{align*}
  \eta_M(s) \eta_M(-s') \wt{S}_{2M}(u\sqrt{2M} + s, u\sqrt{2M} + s')
  &\sim \frac{1}{\pi} \exp\left(-\frac{1}{2}(s-\overline s')^2 -\ip s^2 - \ip
      {s'}^2\right) \\
  & = \frac{1}{\pi} \exp\left( -\frac{|s|^2}{2} - \frac{|s'|^2}{2} +
    s \overline{s}' \right).
\end{align*}

Next, we set $\mathbf{D}$ to be the $2m \times 2m$ diagonal matrix given by 
\[
\mathbf{D}_M = \mathrm{diag}\big( \eta_M(s_1), \eta_M(-s_1), \ldots,
\eta_M(s_m), \eta(-s_M) \big),
\]
noting that $\det \mathbf{D}_M = 1$.  It follows that 
\begin{align*}
& \lim_{M \rightarrow \infty} R_{0,m}(-, \mathbf{z}) = \lim_{M   
  \rightarrow \infty} \Pf \left( \mathbf{D}_M \big[\wt K_{2M}(s_k ,
  s_{k'}) \big]_{k,k'=1}^m \mathbf{D}_M^{-1} \right) \\
& \qquad \qquad = \Pf \begin{bmatrix}
0 & \frac{1}{\pi} \exp\left( -\frac{|s_k|^2}{2} - \frac{|s_{k'}|^2}{2} +
    s_k \overline{s}_{k'} \right) \\
-\frac{1}{\pi} \exp\left( -\frac{|s_k|^2}{2} - \frac{|s_{k'}|^2}{2} +
    s_k \overline{s}_{k'} \right) & 0
\end{bmatrix}_{k,k'=1}^m \\
& \qquad \qquad = \det\left[
\frac{1}{\pi}\exp\left( -\frac{|s_k|^2}{2} - \frac{|s_{k'}|^2}{2} +
  s_{k} \overline{s}_{k'} 
\right) \right]_{k,k'=1}^m,
\end{align*}
where the last equation follows from Section 4.6 of
\cite{sinclair-2007}.   
\end{proof}

\begin{proof}[Proof of Theorem~\ref{thm:5}]
First we consider the case where $s$ and $s'$ are both in the open upper
half plane.  From Theorem~\ref{thm:3}, 
\begin{align*}
&\wt{S}_{2M}(2M)(u\sqrt{2M} + s, u\sqrt{2M} + s') = \frac{ie^{-\frac{1}{2}(s
  - \overline{s}')^2}}{\sqrt{2\pi}}(\overline s' - s)
\sqrt{\erfc\big(\sqrt{2} \ip s \big) \erfc\big(\sqrt{2} \ip{s'} \big)} \\
& \hspace{2cm} \times \exp\big(-2Mu^2 - u \sqrt{2M}(s + \overline s') -
s \overline s'\big) e_M\big(2Mu^2 + u \sqrt{2M}(s + \overline s') +
s\overline s'\big), \\
&\wt{DS}_{2M}(2M)(u\sqrt{2M} + s, u\sqrt{2M} + s') = \frac{e^{-\frac{1}{2}(s
  - s')^2}}{\sqrt{2\pi}}(s' - s)
\sqrt{\erfc\big(\sqrt{2} \ip s \big) \erfc\big(\sqrt{2} \ip{s'} \big)} \\
& \hspace{2cm} \times \exp\big(-2Mu^2 - u \sqrt{2M}(s + s') -
s s'\big) e_M\big(2Mu^2 + u \sqrt{2M}(s + s') + s s'\big),
\end{align*}
and
\begin{align*}
&\wt{IS}_{2M}(2M)(u\sqrt{2M} \!+\! s, u\sqrt{2M} \!+\! s') =
\frac{-e^{-\frac{1}{2}(\overline s
  - \overline{s}')^2}}{\sqrt{2\pi}}(\overline s' - \overline s)
\sqrt{\erfc\big(\sqrt{2} \ip s \big) \erfc\big(\sqrt{2} \ip{s'} \big)} \\
& \hspace{2cm} \times \exp\big(-2Mu^2 - u \sqrt{2M}(\overline s +
\overline s') - \overline s \overline s'\big) e_M\big(2Mu^2 + u
\sqrt{2M}(\overline s + \overline s') + \overline s\overline s'\big).
\end{align*}
By Lemma~\ref{lemma:7}, these converge to the appropriate entries of
the complex/complex kernel $\wt K(s,s')$ as $M \rightarrow \infty$.   

Next we turn to the case where $u$ and $r$ are real, and $s$ is in
the open upper half plane.  In this case, Theorem~\ref{thm:3} yields, 
\begin{align*}
&\wt{S}_{2M}(u\sqrt{2M} + r, u \sqrt{2M} + s) = \frac{i
  e^{-\frac{1}{2}(r - \overline s)^2}}{\sqrt{2\pi}} (\overline s -
r) \sqrt{\erfc\big(\sqrt{2} \ip{s}\big)} \\
& \qquad \quad \times \exp\big(-2Mu^2 - u \sqrt{2M}(r + \overline s) -
r \overline s\big) e_M\big(2Mu^2 + u \sqrt{2M}(r + \overline s) +
r\overline s\big), \\
&\wt{S}_{2M}(u\sqrt{2M} + s, u \sqrt{2M} + r) = \frac{
  e^{-\frac{1}{2}(r - s)^2}}{\sqrt{2\pi}}  \sqrt{\erfc\big(\sqrt{2} \ip{s}\big)} \\
& \qquad \quad \times \exp\big(-2Mu^2 - u \sqrt{2M}(r + s) - r
s\big) e_M\big(2Mu^2 + u \sqrt{2M}(r + s) + r s\big), \\
& \quad + \frac{1}{\sqrt{2\pi}} \sqrt{\erfc\big(\sqrt{2}
  \ip{s} \big)} e^{-\frac{1}{4}(s^2 + \overline s^2)} \exp\big(
-Mu^2 - u\sqrt{2M} \rp{s} \big) \\
& \qquad \quad \times \frac{2^{M-3/2}}{(2M-1)!} \sgn(u \sqrt{2M} + r)
(u\sqrt{2M} + s)^{2M-1} \gamma\left(M - \frac{1}{2}, Mu^2 +
  u\sqrt{2M}r + \frac{r^2}{2} \right),
\end{align*}
and thus, by Lemmas~\ref{lemma:7} and \ref{lemma:8}, 
\[
\lim_{M \rightarrow \infty} \wt{S}_{2M}(u\sqrt{2M} + r, u \sqrt{2M} + s) = \frac{i
  e^{-\frac{1}{2}(r - \overline s)^2}}{\sqrt{2\pi}} (\overline s -
r) \sqrt{\erfc\big(\sqrt{2} \ip{s}\big)},
\]
and
\[
\lim_{M \rightarrow \infty} \wt{S}_{2M}(u\sqrt{2M} + s, u \sqrt{2M} + r) = \frac{
  e^{-\frac{1}{2}(r - s)^2}}{\sqrt{2\pi}}  \sqrt{\erfc\big(\sqrt{2}
  \ip{s}\big)}.
\]
The limiting values for $\wt{DS}_{2M}(u\sqrt{2M} + r, u \sqrt{2M} +
s)$ and $\wt{IS}_{2M}(u\sqrt{2M} + r, u \sqrt{2M} + s)$ follow from
this as well, since 
\[
\wt{DS}_{2M}(z,z') = -i\wt{S}_{2M}(z,\overline z') \qquad \mbox{and}
\qquad 
\wt{IS}_{2M}(z,z') = i \wt{S}_{2M}(\overline z', z). 
\]

Finally, we turn to the case where $u$, $r$ and $r'$ are all real.
Here, Theorem~\ref{thm:3} implies that
\begin{align*}
& \wt{S}_{2M}(u\sqrt{2M} + r, u \sqrt{2M} + r') =
\frac{1}{\sqrt{2\pi}} e^{-\frac{1}{2}(r-r')^2} \\ 
& \quad \quad \times \exp\big(-2Mu^2 - u\sqrt{2M}(r + r') - r r' \big) 
e_M\big(2Mu^2 + u\sqrt{2M}(r + r') + r r' \big) \\
& \quad + \frac{1}{\sqrt{2\pi}}  e^{-\frac{1}{2}r^2} \exp\big(
-Mu^2 - u\sqrt{2M} r \big) \\
& \quad \quad \times \frac{2^{M-3/2}}{(2M-1)!} \sgn(u \sqrt{2M} + r')
(u\sqrt{2M} + r)^{2M-1} \gamma\left(M - \frac{1}{2}, Mu^2 +
  u\sqrt{2M}r' + \frac{r'^2}{2} \right).
\end{align*}
and
\begin{align*}
& \wt{DS}_{2M}(u\sqrt{2M} + r, u \sqrt{2M} + r') = \frac{(r' -
  r)}{\sqrt{2\pi}} e^{-\frac{1}{2}(r-r')^2} \\
& \qquad \quad \times \exp\big(-2Mu^2 - u\sqrt{2M}(r + r') - r r' \big) 
e_M\big(2Mu^2 + u\sqrt{2M}(r + r') + r r' \big),
\end{align*}
From Lemmas~\ref{lemma:7} and \ref{lemma:8}, we see that  
\[
\lim_{M \rightarrow \infty} \wt{S}_{2M}(u\sqrt{2M} + r, u \sqrt{2M} + r') =
\frac{1}{\sqrt{2\pi}} e^{-\frac{1}{2}(r-r')^2},
\]
and 
\[
\lim_{M \rightarrow \infty} \wt{DS}_{2M}(u\sqrt{2M} + r, u \sqrt{2M} +
r') = \frac{(r' - r)}{\sqrt{2\pi}} e^{-\frac{1}{2}(r-r')^2}.
\]

All that remains to show is 
\[
\lim_{M \rightarrow \infty} \wt{IS}_{2M}(u\sqrt{2M} + r, u \sqrt{2M} +
r') = \frac{1}{2} \sgn(r' - r) \erfc\left( \frac{|r -
    r'|}{\sqrt{2}}\right). 
\]
First we write
\begin{align*}
\wt{IS}_{2M}(x,x') &= \frac{e^{-x^2/2}}{2\sqrt{\pi}} \sgn(x')\left\{
  I(x,x') - \int_0^{x'^2/2} \frac{e^{-t}}{\sqrt t} C_M(x\sqrt{2t}) \,
  dt \right\} \\
& \qquad - \frac{e^{-x'^2/2}}{2\sqrt{\pi}} \sgn(x) \left\{
  I(x',x) - \int_0^{x^2/2} \frac{e^{-t}}{\sqrt t} C_M(x'\sqrt{2t}) \,
  dt \right\},
\end{align*}
where $I(x,x')$ is given as in (\ref{eq:35}) and $C_M = \cosh - c_M$.
That is,
\begin{align}
& \wt{IS}_{2M}(x,x') = \frac{e^{-x^2/2}}{2\sqrt{\pi}} \sgn(x')  I(x,x')
- \frac{e^{-x'^2/2}}{2\sqrt{\pi}} \sgn(x)  I(x',x) \nonumber \\
& \qquad + \frac{e^{-x'^2/2}}{2\sqrt{\pi}} \sgn(x) \int_0^{x^2/2} \frac{e^{-t}}{\sqrt
  t} C_M(x'\sqrt{2t}) \,  dt - \frac{e^{-x^2/2}}{2\sqrt{\pi}} \sgn(x') 
\int_0^{x'^2/2} \frac{e^{-t}}{\sqrt t} C_M(x\sqrt{2t}) \,
  dt. \label{eq:36}
\end{align}
Making the substitutions $x = u\sqrt{2M} + r$ and $x' = u\sqrt{2M} +
r'$, and assuming that $M$ is sufficiently large, (\ref{eq:34}) yields 
\[
\lim_{M \rightarrow \infty}\frac{e^{-x^2/2}}{2\sqrt{\pi}} \sgn(x')
I(x,x') =  \left\{
\begin{array}{ll}
\frac{1}{4} \erfc\left(\frac{r - r'}{\sqrt{2}} \right) & \quad
\mbox{if} \; u > 0; \\ & \\
-\frac{1}{4} \erfc\left(\frac{r' - r}{\sqrt{2}} \right) & \quad
\mbox{if} \; u < 0.
\end{array}
\right.
\]
Consequently,
\begin{align*}
& \lim_{M \rightarrow \infty}\left\{\frac{e^{-x^2/2}}{2\sqrt{\pi}} \sgn(x')  I(x,x')
- \frac{e^{-x'^2/2}}{2\sqrt{\pi}} \sgn(x)  I(x',x) \right\} \\
& \hspace{6cm} =
\frac{1}{4} \erfc\left(\frac{r-r'}{\sqrt 2} \right) - \frac{1}{4}
\erfc\left(\frac{r'-r}{\sqrt 2} \right)  \\
& \hspace{6cm} = -\frac{1}{2} \erf\left(\frac{r-r'}{\sqrt 2} \right).
\end{align*}
Thus, if we can show that the second line of (\ref{eq:36}) goes to 0
as $M \rightarrow \infty$, we will have 
\begin{align*}
\lim_{M \rightarrow \infty} \left\{ \frac{1}{2} \sgn(x - x') +
  IS_{2M}(x,x') \right\} &= \frac{1}{2} \sgn(r - r') - \frac{1}{2}
\erf\left(\frac{r-r'}{\sqrt 2} \right) \\ &= \frac{1}{2} \sgn(r - r')
\erfc\left(\frac{|r-r'|}{\sqrt{2}} \right) 
\end{align*}
as desired.  

We thus consider
\[
\frac{e^{-x'^2/2}}{2\sqrt{\pi}} \sgn(x) \int_0^{x^2/2} \frac{e^{-t}}{\sqrt
  t} C_M(x'\sqrt{2t}) \,  dt.
\]
Clearly, for any $v > 0$,
\[
C_M(v) \leq \sum_{m=2M-1}^{\infty} \frac{v^m}{m!} = \frac{\gamma(2M,
  v)}{\Gamma(2M)} \leq 1.
\]
Thus,
\[
\int_0^{x^2/2} \frac{e^{-t}}{\sqrt
  t} C_M(x'\sqrt{2t}) \,  dt \leq \int_0^{x^2/2} \frac{e^{-t}}{\sqrt
  t} \,  dt = \sqrt{\pi} \erf\left(\frac{|x|}{\sqrt{2}}\right) \leq
\sqrt{\pi}.
\]
It follows that
\[
\frac{e^{-x'^2/2}}{2\sqrt{\pi}} \int_0^{x^2/2} \frac{e^{-t}}{\sqrt t}
C_M(x'\sqrt{2t}) \,  dt \leq \frac{e^{-x'^2/2}}{2},
\]
and making the substitution $x' = u\sqrt{2M} + r'$, this goes to 0 as
$M \rightarrow \infty$.
\end{proof}

\appendix

\newcommand{\appsection}[1]{\let\oldthesection\thesection
  \renewcommand{\thesection}{Appendix \oldthesection}
  \section{#1}\let\thesection\oldthesection}

\appsection{Correlation Functions for $\beta=1$ and $\beta=4$ Hermitian
  Ensembles}
\label{app:A}

In this appendix we will use the Pfaffian Cauchy-Binet Formula (see
\ref{app:B}) in order to derive the correlation functions of the
$\beta=1$ and $\beta=4$ Hermitian ensembles.  We will keep the
exposition brief, but will introduce all notation necessary for this
appendix to be read independently from the main body of the paper.  We
reuse much of the notation from main body of the paper so that we may
also reuse the same proofs.  For convenience, $N$ will be a fixed even
integer; similar results are true for odd integers.  

Given a Borel measure $\nu$ on $\R$ we define the associated partition
function to be 
\[
Z^{\nu} := \frac{1}{N!} \int_{\R^N} |\Delta(\bs{\gamma}) | \,
d\nu_N(\bs{\gamma}),
\]
where $\Delta(\bs{\gamma})$ is the Vandermonde determinant in the
variables $\gamma_1, \gamma_2, \ldots, \gamma_N$ and $\nu_N$ is the
product measure of $\nu$ on $\R^N$.  When $\beta=1$ we define the
function $\mathcal{E}: \R^2 \rightarrow \{-\frac{1}{2},0,\frac{1}{2}\}$ and the
operator $\epsilon_{\nu}$ on $L^2(\nu)$ by
\[
\mathcal{E}(\gamma, \gamma') := \frac{1}{2} \sgn(\gamma - \gamma')
\qquad \mbox{and} \qquad
\epsilon_{\nu}g(\gamma) := \int_{\R^2} g(y) \mathcal{E}(y, \gamma) \, d\nu(y).
\]
When $\beta=4$ we define $\mathcal{E}(\gamma, \gamma') := 0$ and
$\epsilon_{\nu}g(y) := g'(y)$.  We use $\epsilon_{\nu}$ to define the
skew-symmetric bilinear form $\la \cdot | \cdot \ra_{\nu}$ on
$L^2(\nu)$ given by
\[
\la g | h \ra_{\nu} := \int_{\R} \big( g(\gamma) \epsilon_{\nu} h(\gamma) -
\epsilon_{\nu}g(\gamma) h(\gamma) \big) d\nu(\gamma).
\]
\begin{thm}
Let $b := \sqrt{\beta}$, and let $\mathbf{q}$ be a family of $Nb$ monic
polynomials such that $\deg q_n = n$.  Then,
\[
Z^{\nu} = \frac{(Nb)!}{N!} \Pf \mathbf{U}_{\mathbf{q}}^{\nu},
\]
where $\mathbf{U}_{\mathbf{q}}^{\nu} = [ \la q_n | q_{n'} \ra_{\nu} ];
\quad n,n'=0,1,\ldots,Nb-1$.  
\end{thm}
This theorem follows from de Bruijn's identities \cite{MR0079647}.  

We set $\lambda$ to be Lebesgue measure on $\R$.  If there is some
Borel measurable function $w: \R \rightarrow [0,\infty)$ so that $\nu
= w\lambda$ (that is, $d\nu/d\lambda = w$) then we define $Z^w :=
Z^{\nu}$.  Clearly,
\[
Z^w = \frac{1}{N!} \int_{\R^N} \Omega_N(\bs{\gamma}) \,
d\lambda_N(\bs{\gamma}) \qquad \mbox{where} \qquad
\Omega_N(\bs{\gamma}) := \bigg\{\prod_{n=1}^N w(\gamma_n) \bigg\}
|\Delta(\bs{\gamma})|^{\beta}.
\]
We may specify an ensemble of Hermitian matrices by demanding that
its joint probability density function is given by $\Omega_N$.  The
$n$th correlation function of this ensemble is then defined to be
$R_n: \R^n \rightarrow [0,\infty)$ where
\begin{equation}
\label{eq:13}
R_n(\mathbf{y}) := \frac{1}{Z^w} \cdot \frac{1}{(N-n)!}
\int_{\R^{N-n}} \Omega_N(\mathbf{y} \vee \bs{\gamma}) \,
d\lambda_{N-n}(\bs{\gamma}),
\end{equation}
where $\mathbf{y} \vee \bs{\gamma} \in \R^N$ is the vector formed by
concatenating the vectors $\mathbf{y} \in \R^n$ and $\mathbf{\gamma}
\in \R^{N-n}$.  By definition, $R_0 = 1$.  Here we take (\ref{eq:13})
as the definition of the $n$th correlation function; one can use
the point process formalism to show that this definition is consistent
with the definition derived in that manner.  See \cite{MR2185737} for
details.  

We set $\mu_{n,n'}$ to be the $n,n'$ entry of
$(\mathbf{U}_{\mathbf{q}}^{w\lambda})^{-\transpose}$, and we define
$\widetilde{q_n} := w q_n$.   Using this notation we define the
functions $S_N, IS_N$ and $DS_N : \R^2 \rightarrow \R$ by
\[
S_N(\gamma, \gamma') := \frac{2}{b} \sum_{n,n'=0}^{Nb-1} \!\!
\mu_{n,n'} \, \widetilde{q}_n(\gamma) \, \epsilon_{\lambda}
\widetilde{q}_{n'}(\gamma'),  \qquad
IS_N (\gamma, \gamma') := \frac{2}{b} \sum_{n,n'=0}^{Nb-1} \!\!
\mu_{n,n'} \, \epsilon_{\lambda} \widetilde{q}_n(\gamma) \, 
\epsilon_{\lambda} \widetilde{q}_{n'}(\gamma')
\]
and
\[
DS_N (\gamma, \gamma') := \frac{2}{b} \sum_{n,n'=0}^{Nb-1} 
\mu_{n,n'} \, \widetilde{q}_n(\gamma) \, \widetilde{q}_{n'}(\gamma').
\]
The matrix kernel of our ensemble is then defined to be
\[
K_N(\gamma, \gamma') := \begin{bmatrix}
DS_N(\gamma, \gamma') & S_N(\gamma, \gamma') \\
-S_N(\gamma', \gamma) & IS_N(\gamma, \gamma') + \mathcal{E}(\gamma,
\gamma'). 
\end{bmatrix}
\]

\begin{thm}
\label{thm:4}
\[
R_n(\mathbf{y}) = \Pf\big[ K_N(\mathbf{y}_j, \mathbf{y_{j'}} ) \big];
\qquad j,j'=1,2,\ldots,n.
\]
\end{thm}

Our proof of this theorem begins by setting $\eta$ to be the measure
on $\R$ given by
\[
d \eta(\gamma) = \sum_{t=1}^T c_t \, d \delta(\gamma - y_t),
\]
where $y_1, y_2, \ldots, y_T$ are real numbers and $c_1, c_2, \ldots,
c_T$ are indeterminants and $\delta$ is the probability measure on
$\R$ with point mass at 0.  We will assume that $T \geq N$.  As with
Theorem~\ref{thm:2} in the main body of this paper, we will prove
Theorem~\ref{thm:4} by expanding $Z^{w(\lambda + \eta)}/Z^w$ in two
different ways and then equating the coefficients of certain products
of $c_1, c_2, \ldots, c_T$.  
\begin{lemma}
\label{lemma:1}
\[
\frac{Z^{w(\lambda + \eta)}}{Z^w} = \Pf \big( \mathbf{J} + \big[ \sqrt{c_t
  c_{t'}} K_N(y_t, y_{t'}) \big] \big); \qquad t,t'=1,2,\ldots,T,
\]
where $\mathbf{J}$ is defined to be the $2T \times 2T$ matrix
consisting of $2 \times 2$ blocks given by
\[
\mathbf{J} := \left[
\delta_{t,t'} \begin{bmatrix}
0 & 1 \\
-1 & 0 
\end{bmatrix}
\right];
\qquad t,t = 1,2,\ldots, T.
\]
\end{lemma}
Lemma~\ref{lemma:1} is proved in Section~\ref{sec:proof-lemma-refl}.  

For each $N \geq 0$ we define $\mf{I}_n^T$ to be the
set of increasing functions from $\{1,2,\ldots,n\}$ into
$\{1,2,\ldots,T\}$.  Given a vector $\mathbf{y} \in \R^T$ and an element $\mf{t} \in
\mf{I}_n^T$, we define the vector $\mathbf{y}_{\mf{t}} \in \R^n$ by
$\mathbf{y}_{\mf{t}} = \{y_{\mf{t}(1)}, y_{\mf{t}(2)}, \ldots,
y_{\mf{t}(n)}\}$.  
\begin{lemma}
\label{lemma:2}
\begin{equation}
\label{eq:15}
\frac{Z^{w(\lambda + \eta)}}{Z^w} = 1 + \sum_{n=1}^N \sum_{\mf{t} \in
  \mf{I}_n^T} \bigg\{ \prod_{j=1}^T c_{\mathbf{t}(j)} \bigg\}
  R_n(\mathbf{y}_{\mf{t}}).
\end{equation}
\end{lemma}
The proof of Lemma~\ref{lemma:2} is given in
Section~\ref{sec:proof-lemma-refl-1}.  

Finally, we set $\mathbf{K}$ to be the $2T \times 2T$ block matrix
given by 
\[
\mathbf{K} := \big[ \sqrt{c_t c_{t'}} \, K_N(y_t, y_{t'}) \big]; \qquad
t,t'=1,2,\ldots, T.
\]
From the formula for the Pfaffian of the sum of two antisymmetric
matrices (see Lemma~\ref{lemma:6}) and Lemma~\ref{lemma:1} we have that
\begin{equation}
\label{eq:16}
\frac{Z^{w(\lambda + \eta)}}{Z^w} = \Pf[ \mathbf{J} + \mathbf{K} ] = 1
+ \sum_{n=1}^T \sum_{\mf{t} \in  \mf{I}_n^T} \Pf \mathbf{K}_{\mf{t}},
\end{equation}
where for each $\mf{t} \in \mf{I}_n^T$, $\mathbf{K}_{\mf{t}}$ is the
$2n \times 2n$ antisymmetric matrix given by
\[
\mathbf{K}_{\mf{t}} = [ \sqrt{c_{\mf{t}(j)} c_{\mf{t}(j')}}K_N(y_{\mf{t}(j)}, y_{\mf{t}(j')})]; \qquad j,j'=1,2,\ldots,S.
\]
Finally,
\[
\Pf \mathbf{K}_{\mf{t}} = \bigg\{ \prod_{j=1}^n c_{\mf{t}(j)} \bigg\}
\Pf[ K_N(y_{\mf{t}(j)}, y_{\mf{t}(j')}) ]; \qquad j,j'=1,2,\ldots,n,
\]
and Theorem~\ref{thm:4} follows from Lemma~\ref{lemma:2} by comparing
coefficients of $c_1 c_2 \cdots c_n$ in (\ref{eq:15}) and (\ref{eq:16}).

\appsection{The Pfaffian Cauchy-Binet Formula}
\label{app:B}

\begin{thm}[Rains]
Suppose $\mathbf{B}$ and $\mathbf{C}$ are respectively $2J \times 2J$
and $2K \times 2K$ antisymmetric matrices with non-zero Pfaffians.
Then, given any $2J \times 2K$ matrix $\mathbf{A}$,
\[
\frac{\Pf(
  \mathbf{C}^{-\mathsf{T}} - \mathbf{A}^{\mathsf{T}} \mathbf{B A} )}{\Pf(
  \mathbf{C}^{-\mathsf{T}})} =  \frac{\Pf(
  \mathbf{B}^{-\mathsf{T}} - \mathbf{A C A}^{\mathsf{T}} )}{\Pf(
  \mathbf{B}^{-\mathsf{T}})}.
\]
\end{thm}
\begin{proof}
Let $\mathbf{I}_{2K}$ and $\mathbf{I}_{2J}$ be respectively the $2K
  \times 2K$ and $2J \times 2J$ identity matrices, and let
  $\mathbf{O}$ be the $2J \times 2K$ matrix whose entries are all 0.
  Then, an easy calculation shows that
\[
\begin{bmatrix}
\mathbf{I}_{2J} & \mathbf{O} \\
\mathbf{A}^{\transpose} \mathbf{B} & \mathbf{I}_{2K} 
\end{bmatrix}
\begin{bmatrix}
\mathbf{B}^{-\transpose} & -\mathbf{A} \\
\mathbf{A}^{\transpose} & \mathbf{C}^{-\transpose} 
\end{bmatrix}
\begin{bmatrix}
\mathbf{I}_{2J} & -\mathbf{B A} \\
\mathbf{O}^{\transpose} & \mathbf{I}_{2K} 
\end{bmatrix}
=
\begin{bmatrix}
\mathbf{B}^{-\transpose} & \mathbf{O} \\
\mathbf{O}^{\transpose} & \mathbf{C}^{-\transpose} -
\mathbf{A}^{\transpose} \mathbf{B A} 
\end{bmatrix},
\]
and similarly
\[
\begin{bmatrix}
\mathbf{I}_{2J} & -\mathbf{AC} \\
\mathbf{O}^{\transpose} & \mathbf{I}_{2K} 
\end{bmatrix}
\begin{bmatrix}
\mathbf{B}^{-\transpose} & -\mathbf{A} \\
\mathbf{A}^{\transpose} & \mathbf{C}^{-\transpose} 
\end{bmatrix}
\begin{bmatrix}
\mathbf{I}_{2J} & \mathbf{O} \\
\mathbf{C} \mathbf{A}^{\transpose} & \mathbf{I}_{2K} 
\end{bmatrix}
=
\begin{bmatrix}
\mathbf{B}^{-\transpose} -
\mathbf{A C A}^{\transpose}  & \mathbf{O} \\
\mathbf{O}^{\transpose} & \mathbf{C}^{-\transpose} 
\end{bmatrix}.
\]
Now, if $\mathbf{D}$ and $\mathbf{E}$ are $2N \times 2N$ matrices and
$\mathbf{D}$ is antisymmetric, then it is well known that
$\Pf(\mathbf{E D E}^{\transpose}) = \Pf \mathbf{D} \cdot \det
\mathbf{E}$.  From which we conclude that 
\[
\Pf \begin{bmatrix}
\mathbf{B}^{-\transpose} & \mathbf{O} \\
\mathbf{O}^{\transpose} & \mathbf{C}^{-\transpose} -
\mathbf{A}^{\transpose} \mathbf{B A} 
\end{bmatrix} = \Pf\begin{bmatrix}
\mathbf{B}^{-\transpose} -
\mathbf{A C A}^{\transpose}  & \mathbf{O} \\
\mathbf{O}^{\transpose} & \mathbf{C}^{-\transpose} 
\end{bmatrix}.
\]
The theorem follows since the Pfaffian of the direct sum of two even
rank antisymmetric matrices is the product of the Pfaffians of the two
matrices.  That is
\[
\Pf(\mathbf{B}^{-\transpose}) \Pf(\mathbf{C}^{-\transpose} -
\mathbf{A}^{\transpose} \mathbf{B A}) = \Pf(\mathbf{B}^{-\transpose} -
\mathbf{A C A}^{\transpose})\Pf(\mathbf{C}^{-\transpose} ). \qedhere
\]

\end{proof}

\bibliography{bibliography}

\noindent\rule{4cm}{.5pt}

\vspace{.25cm}
\noindent {\sc \small Alexei Borodin}\\
{\small California Institute of Technology \\
Pasadena, California \\
email: {\tt borodin@caltech.edu}}

\vspace{.25cm}
\noindent {\sc \small Christopher D.~Sinclair}\\
{\small Max Planck Institut f\"ur Mathematik \\
Bonn, Germany \\
email: {\tt christopher.sinclair@colorado.edu}}

\end{document}

%% file: preamble.tex
\usepackage{amsmath,amsthm,amsfonts,amscd,mathrsfs,pifont} 
\usepackage{eucal}  
\usepackage{verbatim}     
\usepackage[dvips]{graphicx}
\usepackage{hyperref}

\newtheorem{thm}{Theorem}[section]
\newtheorem{cor}[thm]{Corollary}

\newtheorem{lemma}[thm]{Lemma}

\newtheorem{thm*}{Theorem}[]
\newtheorem{cor*}[thm*]{Corollary}
\newtheorem{claim*}[thm*]{Claim}
\newtheorem{lemma*}[thm*]{Lemma}
\newtheorem{prop*}[thm*]{Proposition}
\newtheorem{conj*}[thm*]{Conjecture}

\theoremstyle{definition}

\newtheorem{question*}{Question}[section]

\newtheorem{defn*}{Definition}

\theoremstyle{remark}
\newtheorem*{rem}{Remark}

\newcommand{\BB}[1]{\ensuremath{\mathbb{#1}}}

\newcommand{\R}{\ensuremath{\BB{R}}}
\newcommand{\Z}{\ensuremath{\BB{Z}}}
\newcommand{\C}{\ensuremath{\BB{C}}}

\newcommand{\bs}{\ensuremath{\boldsymbol}}
\newcommand{\mf}{\ensuremath{\mathfrak}}
\newcommand{\wt}{\ensuremath{\widetilde}}

\newcommand{\la}{\ensuremath{\langle}}
\newcommand{\ra}{\ensuremath{\rangle}}

\newcommand{\CnoR}{\ensuremath{\C_{\ast}}}
\newcommand{\CnoRsub}{\ensuremath{\C_{\ast}}}
\newcommand{\transpose}{\ensuremath{\mathsf{T}}}

\newcommand{\ip}[1]{\mathrm{Im}(#1)}
\newcommand{\rp}[1]{\mathrm{Re}(#1)}

\DeclareMathOperator{\sgn}{sgn}

\DeclareMathOperator{\erf}{erf}
\DeclareMathOperator{\erfc}{erfc}

\DeclareMathOperator{\Tr}{Tr}

\DeclareMathOperator{\Pf}{Pf}